\documentclass[lettersize,journal]{IEEEtran}

\usepackage{amsmath,amsfonts}
\usepackage{amssymb,mathtools}
\usepackage{graphicx}
\usepackage{cite}
\usepackage{booktabs}
\usepackage[linesnumbered,ruled,vlined]{algorithm2e}
\usepackage{array}
\usepackage[caption=false,font=normalsize,labelfont=sf,textfont=sf]{subfig}
\usepackage{textcomp}
\usepackage{stfloats}
\usepackage[normalem]{ulem}
\usepackage{xcolor}
\usepackage{placeins}
\usepackage{url}
\usepackage[hidelinks]{hyperref}
\hypersetup{
    pdftitle={Reliability- and Connectivity-Constrained Age-of-Information Optimization for UAV Swarm IoT Data Collection},
    pdfauthor={Milad Bafarassat and Sinem Coleri},
    pdfsubject={IEEE Internet of Things Journal manuscript},
    pdfkeywords={UAV swarm, Internet of Things, age of information, finite-blocklength communications, successive convex approximation}
}
\makeatother
\usepackage{orcidlink}
\usepackage{balance}
\usepackage{xcolor}
\usepackage{stfloats}
\newtheorem{proposition}{Proposition}
\newtheorem{lemma}{Lemma}

\newcommand{\rev}[1]{\textcolor{black}{#1}}
\usepackage[normalem]{ulem}
\newcommand{\revii}[1]{\textcolor{black}{#1}}
\newcommand{\czo}{c_0} 

\newcommand{\vn}{\mathbf{v}_n}
\newcommand{\qn}{\mathbf{q}_n}
\newcommand{\Lbar}{\bar{\mathbf{L}}}
\newcommand{\vbar}{\bar{\mathbf{v}}}

\begin{document}

\title{Reliability- and Connectivity-Constrained Age-of-Information Optimization for UAV Swarm IoT Data Collection}

\author{Milad~Bafarassat\orcidlink{0009-0004-4461-9347}
and~Sinem~Coleri\orcidlink{0000-0002-7502-3122},~\IEEEmembership{Fellow,~IEEE}%
\thanks{M.~Bafarassat and S.~Coleri are with the Wireless Networks Laboratory, Department of Electrical and Electronics Engineering, Ko\c{c} University, Istanbul, T\"urkiye (e-mail: mbafarassat26@ku.edu.tr; scoleri@ku.edu.tr).

Sinem Coleri acknowledges the support of TUBITAK (Scientific and Technological Research Council of Türkiye) Science Award grant.

Copyright (c) 2026 IEEE. Personal use of this material is permitted. However, permission to use this material for any other purposes must be obtained from the IEEE by sending a request to pubs-permissions@ieee.org.}%
}

\markboth{IEEE Internet of Things Journal}%
{Bafarassat and Coleri: Freshness, Reliability and Connectivity Optimization for UAV Swarms}

\maketitle

\begin{abstract}
Uncrewed aerial vehicle~(UAV) swarms provide a flexible platform for Internet of Things~(IoT) data collection by gathering delay-sensitive measurements from distributed sensor clusters and relaying them to a fusion center~(FC). Between successive updates, the FC estimates the current value of each monitored quantity; however, the estimation accuracy degrades as the last update ages. Such missions require jointly optimizing estimation accuracy, short-packet communication reliability, and swarm connectivity, yet existing approaches typically prioritize freshness while neglecting its impact on connectivity. This paper proposes a topology-coupled urgency~(TCU) scheduler that minimizes the FC's accumulated estimation error while incorporating the swarm's expected algebraic connectivity into the scheduling objective. The resulting connectivity reward encourages idle UAVs to reposition so as to preserve swarm connectivity. For each time slot, the resulting problem is a mixed-integer nonconvex program, which we solve by successive convex approximation. Specifically, the binary assignments are relaxed, while the connectivity reward is represented through a linear matrix inequality on the expected Laplacian. The finite-blocklength rate and the inter-UAV link reliability are linearized in the same convex subproblem, thereby preserving their coupling rather than separating them into alternating blocks. The integer assignments are subsequently recovered by Hungarian matching and verified against the original nonconvex constraints, while a deficit mechanism ensures long-term cluster coverage without requiring explicit long-horizon planning. In simulations comprising more than $19{,}000$ transmission attempts, all delivered updates satisfy the reliability target. Against the best age-of-information (AoI) heuristic, the proposed scheduler matches mean AoI while reducing the FC's estimation error by $12\%$ and raising the delivery success rate by $25$ percentage points. An ablation study further confirms the role of the connectivity reward: removing it in a dispersed deployment degrades mean AoI by $4.8$ times.
\end{abstract}


\section{Introduction}
\label{sec:intro}
\IEEEPARstart{U}{ncrewed} aerial vehicle~(UAV) swarms provide flexible, on-demand wireless infrastructure for Internet of Things~(IoT) applications~\cite{Yuan2025LAWN,Mozaffari2019}. In a typical data-collection mission, a swarm collects delay-sensitive measurements from distributed ground sensors over short-packet uplinks and relays them to a fusion center~(FC), while inter-UAV links maintain swarm connectivity. \rev{Because sensors send updates intermittently, the FC must estimate the current value of each monitored process between consecutive successful deliveries. The quality of this estimate deteriorates as the most recently received update ages, at a rate governed by the dynamics of the underlying process. System performance is therefore determined by three tightly coupled factors: estimation accuracy,} short-packet communication reliability and swarm connectivity. Optimizing any one objective in isolation inevitably compromises the others. \rev{An age-of-information (AoI)-based scheduler treats sources with different dynamics as equally urgent at the same age and may therefore under-serve rapidly varying processes. A reliability-centric scheduler may direct UAVs toward favorable access links at the expense of inter-UAV connectivity, whereas a topology-centric policy may preserve the swarm graph while neglecting the clusters whose updates are most valuable.}

UAV-based IoT data collection spans three complementary lines of work~\cite{Wei2022Survey}\rev{: maximizing communication throughput or sensing coverage, improving energy efficiency or enabling wireless-powered data collection, and minimizing AoI to enhance information freshness.} Within the AoI-driven literature, UAV trajectories and sensor scheduling are jointly optimized through various methods: neural combinatorial optimization~\cite{Wu2022AI}, Transformer-based tour planning~\cite{ZhuAoI2023}, mobility-graph trajectories with mixing-time guarantees~\cite{Tripathi2019Graphs}, energy- and capacity-aware association~\cite{Gao2023AoI} and Lyapunov drift-plus-penalty decomposition solved by block coordinate descent~(BCD)~\cite{LongTVT2024}. Although these approaches differ in their optimization techniques, they share a common objective of minimizing AoI subject to mobility and energy constraints. Consequently, they generally assume homogeneous sensing priorities and overlook three factors that critically influence practical UAV-swarm operation: source-aware sampling, finite-blocklength (FBL) reliability of the sensor data links, and swarm connectivity.

These missing aspects have each received considerable attention, but largely in separate research communities. Remote-estimation studies show that the optimal sampling policy for an Ornstein--Uhlenbeck~(OU) process depends on source dynamics rather than AoI alone~\cite{Ornee2021}. \rev{Similarly, weighted}-AoI scheduling for correlated sources can achieve estimation performance within a constant factor of the optimum~\cite{Ramakanth2024}. \rev{In the setting considered in this paper, each sensor cluster monitors an OU process; hence, the value of serving a cluster is determined by the reduction in the FC's estimation error resulting from a successful update. At the physical layer,} the finite-blocklength~(FBL) normal approximation~\cite{Polyanskiy2010} enables accurate modeling of short-packet communication reliability and is applied to multi-UAV IoT data collection under imperfect channel state information (CSI)~\cite{Cheng2026}. On the network-topology side, algebraic connectivity~\cite{Fiedler1973,Chung1997} forms the basis of connectivity control in mobile robotic networks~\cite{Zavlanos2011}, while matrix concentration tools~\cite{Tropp2015} provide deterministic conditions for maintaining connectivity in random graphs. However, none combines source-aware sampling, finite-blocklength reliability and connectivity within a single UAV swarm scheduler.

\rev{To bridge this gap, we propose a topology-coupled urgency (TCU) scheduling framework that jointly minimizes the accumulated estimation error at the FC and preserves the connectivity of the UAV swarm.} 
The main contributions are summarized as follows:

\begin{itemize}
\item \rev{We propose, for the first time, a scheduler that jointly optimizes three coupled objectives: estimation accuracy, short-packet communication reliability and swarm connectivity. In each slot, the framework jointly determines the UAV-to-cluster assignments, UAV movements, transmission durations, transmit powers, and collected data volumes, subject to assignment, mobility, energy, communication-reliability, and topology constraints.}
\item \rev{We derive the urgency of each sensor cluster directly from the remote-estimation objective as the reduction in the FC's estimation error resulting from a successful delivery. Unlike conventional AoI-based priorities, this urgency metric accounts for both the staleness of the available information and the dynamics of the monitored process. We further incorporate a weighted connectivity reward based on the algebraic connectivity of the expected Laplacian. This reward encourages idle UAVs to reposition in support of swarm connectivity.}
\item \rev{We develop a successive convex approximation (SCA)-based solution to the resulting mixed-integer nonconvex optimization problem. At each SCA iteration, all continuous decision variables are optimized jointly through a single convex subproblem. The connectivity reward is represented exactly using a linear matrix inequality involving the expected Laplacian, while the FBL rate and inter-UAV link-reliability expressions are linearized within the same subproblem. Their common dependence on UAV positions is therefore retained rather than separated across alternating optimization blocks. Binary UAV--cluster assignments are subsequently recovered using Hungarian matching and verified against the original nonconvex constraints. In addition, a deficit-queue mechanism increases the priority of clusters that fall behind their service targets, thereby preventing persistent starvation without requiring explicit long-horizon planning.}
\item We demonstrate the effectiveness of the proposed framework through extensive simulations under \rev{three representative deployment geometries: a uniformly distributed deployment, a load-imbalanced deployment, and a widely dispersed deployment. All methods are evaluated under identical, reliability-certified physical-layer conditions}. Compared with alternating-block optimization, drift-plus-penalty and heuristic baselines, the proposed scheduler \rev{consistently improves estimation accuracy and delivery reliability while maintaining competitive mean AoI. An ablation study further confirms that the connectivity reward is essential for preserving both swarm connectivity and freshness in challenging deployments}.
\end{itemize}
The rest of the paper is organized as follows. Section~\ref{sec:system} describes the system model and assumptions. Section~\ref{sec:formulation} formulates the joint optimization problem. Section~\ref{sec:solution} presents the SCA decomposition and convex surrogates. Section~\ref{sec:results} reports simulation results and the complexity analysis. Section~\ref{sec:conclusion} concludes the paper and outlines future work.

\section{System Model and Assumptions}
\label{sec:system}

\begin{figure}[!t]
\centering
\includegraphics[width=1\columnwidth,height=6cm]{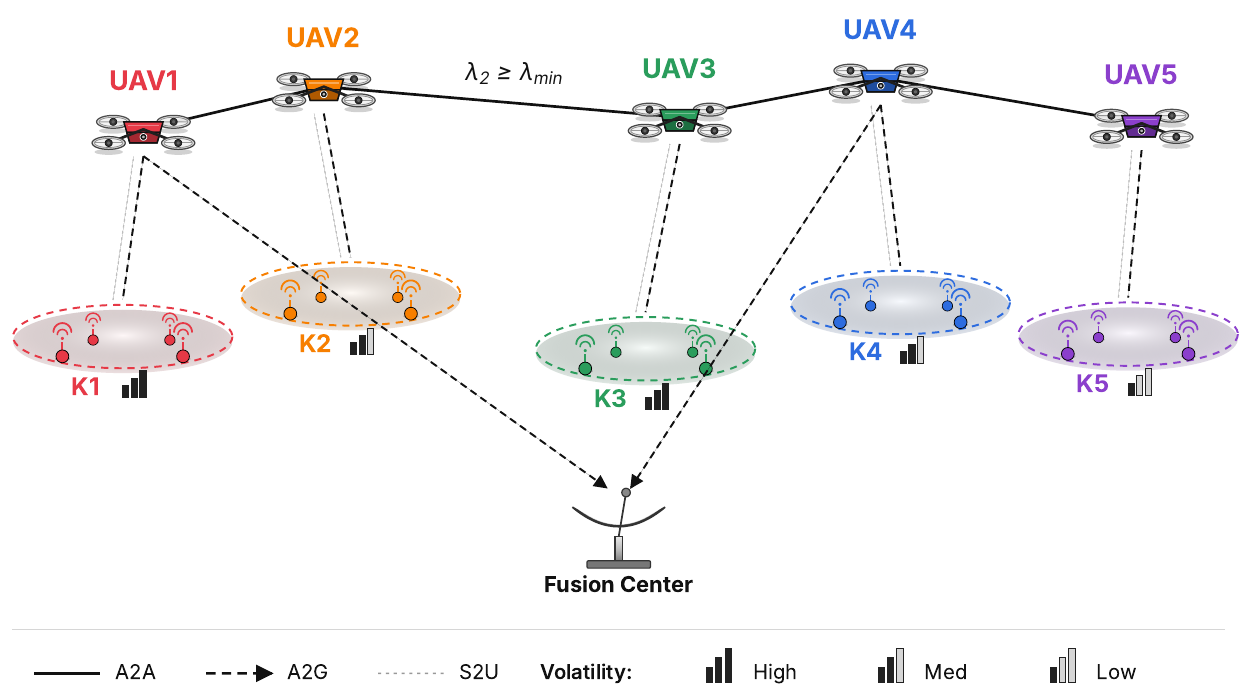}
\caption{UAV-assisted IoT data collection: $N$ UAVs collect from $K$ clusters over S2U uplinks and forward to the FC over A2G links; A2A links maintain connectivity.}
\label{fig:system}
\end{figure}

A swarm of $N$ rotary-wing UAVs collects data from $K$ ground IoT sensor clusters and forwards it directly to an FC, as illustrated in Fig.~\ref{fig:system}. The mission spans $T$ time slots of duration $\delta_t$ seconds. The UAVs operate under a Fly-to-Sense-and-Forward (FSF) protocol~\cite{LongTVT2024}: in slot~$t$, UAV~$n$ starts at position $\qn(t)\in\mathbb{R}^3$, moves at velocity $\vn(t)\in\mathbb{R}^3$, \rev{collects data from} its assigned cluster for $\tau_{\mathrm{s},n}(t)$ seconds and forwards \rev{the collected data to} the FC for $\tau_{\mathrm{f},n}(t)$ seconds, with $\tau_{\mathrm{s},n}(t)+\tau_{\mathrm{f},n}(t)\le\delta_t$. When active, its sensing and forwarding windows are floored at $\tau_s^{\min}$ and $\tau_f^{\min}$ to rule out zero-duration allocations. The slot far exceeds the channel coherence time, so large-scale path gains are evaluated at the end-of-slot position $\qn(t{+}1)=\qn(t)+\vn(t)\,\delta_t$, with small-scale fading and link-reliability statistics applied to these reference gains.

\subsection{Network Architecture and Channel Model}

\begin{table}[!t]
\centering
\caption{Main notation.}
\label{tab:notation}
\footnotesize
\renewcommand{\arraystretch}{1.05}
\setlength{\tabcolsep}{4pt}
\begin{tabular}{@{}ll@{}}
\toprule
Symbol & Meaning \\
\midrule
$N$, $K$, $T$, $\delta_t$ & UAVs, clusters, horizon (slots), slot length \\
$\qn(t)$, $\vn(t)$ & position and velocity of UAV~$n$ \\
$\alpha_{k,n}(t)$ & cluster--UAV assignment indicator \\
$S_k(t)$, $\widetilde S_k(t)$ & attempted and successful service of cluster~$k$ \\
$d_{k,n}(t)$ & delivered data volume (bits) \\
$\tau_{\mathrm{s},n}(t)$, $\tau_{\mathrm{f},n}(t)$ & sensing and forwarding durations \\
$P_k(t)$, $P^{\mathrm{relay}}(t)$ & sensor and relay transmit powers \\
$\hat h_{k,n}(t)$, $\rho_{k,n}(t)$ & pilot channel estimate, Jakes correlation \\
$\gamma_{k,n}^{\mathrm{est}}(t)$ & effective S2U SINR \\
\rev{$\varepsilon_{\max}$, $\varepsilon_{\mathrm{p}}$} & \rev{block-error target, per-packet budget} \\
\revii{$L^{\mathrm{s}}_{k,n}(t)$}\rev{, $n_{\mathrm{p}}$} & \rev{symbols per window, packet length} \\
$\Delta_k(t)$ & age of information of cluster~$k$ \\
$\sigma_{\xi,k}^2(\Delta)$ & OU conditional variance at age $\Delta$ \\
$g_k(t)$ & service gain of cluster~$k$, \eqref{eq:gain} \\
$Z_k(t)$, $\omega_k$ & deficit queue, target service rate \\
$B_k^{\min}(t)$, $B_0$ & rate--distortion and minimum payloads \\
$\tilde p_{nm}(t)$, $\mathcal E(t)$, $R_{\mathrm{c}}$ & A2A link reliability, candidate set, radius \\
$\Lbar(t)$, $\bar\lambda_2(t)$ & expected Laplacian, algebraic connectivity \\
$\vbar_2(t)$ & Fiedler vector of $\Lbar(t)$ \\
$\eta$, $\delta_c$, $\lambda_{\min}$ & per-edge slack, confidence, connectivity floor \\
$\mu$, $\lambda_Z$ & connectivity price, coverage-pressure weight \\
$E_n^{\mathrm{rem}}(t)$, $E_k^{\mathrm{s,rem}}(t)$ & remaining UAV and cluster energies \\
$W_{k,n}(t)$ & matching weight, \eqref{eq:assignment_weights} \\
\bottomrule
\end{tabular}
\end{table}

The network is single-hop: each UAV forwards \rev{the collected data directly} to the FC over the air-to-ground~(A2G) link \rev{within the same} slot, while air-to-air~(A2A) links \rev{support inter-UAV coordination and maintain swarm} connectivity. A single relay power $P^{\mathrm{relay}}(t)$ serves both A2A and A2G. The sensor transmit power $P_k(t)$ of cluster~$k$ and the aerial relay power are bounded by hardware and regulatory limits as $0 \leq P_k(t) \leq P_k^{\max}$ and $0 \leq P^{\mathrm{relay}}(t) \leq P^{\mathrm{relay}}_{\max}$.

Each cluster shares a common sensor-to-UAV~(S2U) sub-band and a single \rev{aggregate} energy budget. \rev{A cluster is successfully served only if its update is both decoded at the collecting UAV and forwarded to the FC within the same time slot.} We assume the cluster head (CH) of each sensor cluster is predetermined and is the only transmitter when the cluster is served. The CH selection algorithm, e.g., LEACH~\cite{Heinzelman2002}, is beyond the scope of this work.

We adopt a one-to-one assignment rule: per slot, each UAV serves at most one cluster and each cluster is assigned to at most one UAV, so each active UAV~$n$ delivers to the FC a single payload $d_{k,n}(t)$ in bits collected from cluster~$k$. Let \(\alpha_{k,n}(t)\in\{0,1\}\) indicate whether cluster~$k$ is assigned to UAV~$n$. The attempted-service indicator for cluster~$k$ is
$S_k(t)=\sum_{m=1}^{N}\alpha_{k,m}(t)$,
so at most $\min\{K,N\}$ S2U transmissions occur in parallel.

Let $h_{k,n}(t)$ denote the channel coefficient of cluster~$k$'s active CH toward UAV~$n$. The S2U signal-to-interference-plus-noise ratio (SINR) at UAV~$n$ is $\gamma_{k,n}^{\mathrm{S2U}}(t)=P_k(t)|h_{k,n}(t)|^2/[\sigma_n^2+\sum_{k'\neq k}S_{k'}(t)P_{k'}(t)|h_{k',n}(t)|^2]$, where $\sigma_n^2$ is the thermal noise at UAV~$n$ and the summation captures co-channel interference from other simultaneously transmitting clusters. The S2U channel is Rayleigh with distance-dependent mean power, $\mathbb{E}[|h_{k,n}(t)|^2]=\beta_0^{\mathrm{S2U}}\|\qn(t{+}1)-\mathbf{q}_k\|^{-\alpha_p^{\mathrm{S2U}}}$, where $\mathbf{q}_k$ is the cluster-head position; the ground uplink propagates through clutter and therefore carries its own reference gain $\beta_0^{\mathrm{S2U}}$ and exponent $\alpha_p^{\mathrm{S2U}}$.

A2A and A2G links operate on orthogonal, interference-free sub-bands (FDMA)~\cite{LiXu2018}, with signal-to-noise ratio~(SNR):
\begin{equation}
\begin{aligned}
\gamma_{nm}^{\mathrm{A2A}}(t) &= \frac{P^{\mathrm{relay}}(t)\,\beta_{nm}^{\mathrm{A2A}}(t)}{\sigma_m^2}, \\
\gamma_n^{\mathrm{A2G}}(t) &= \frac{P^{\mathrm{relay}}(t)\,\beta_{n,\mathrm{FC}}(t)}{\sigma_{\mathrm{FC}}^2},
\end{aligned}
\label{eq:A2A_A2G_SNR}
\end{equation}
where $\sigma_m^2$ and $\sigma_{\mathrm{FC}}^2$ are the noise powers at UAV~$m$ and the FC, and $\beta_{nm}^{\mathrm{A2A}}(t),\beta_{n,\mathrm{FC}}(t)$ the corresponding A2A and A2G path gains. Each noise power is the product of the noise spectral density $N_0$ and the corresponding sub-band bandwidth, $\sigma_n^2=N_0 B^{\mathrm{s}}$, $\sigma_m^2=N_0 B^{\mathrm{A2A}}$ and $\sigma_{\mathrm{FC}}^2=N_0 B^{\mathrm{A2G}}$.

Defining $\boldsymbol{\delta}_{nm}(t)=\qn(t{+}1)-\mathbf{q}_m(t{+}1)$, the path gains are modeled as~\cite{LiXu2018}
\begin{equation}
\begin{aligned}
\beta_{nm}^{\mathrm{A2A}}(t) &= \beta_0^{\mathrm{air}}\|\boldsymbol{\delta}_{nm}(t)\|^{-\alpha_p^{\mathrm{air}}}, \\
\beta_{n,\mathrm{FC}}(t) &= \beta_0^{\mathrm{air}}\|\qn(t{+}1)-\mathbf{q}_{\mathrm{FC}}\|^{-\alpha_p^{\mathrm{air}}},
\end{aligned}
\label{eq:path_gains}
\end{equation}
where $\mathbf{q}_{\mathrm{FC}}$ denotes the FC's fixed position and $\beta_0^{\mathrm{air}}$ and $\alpha_p^{\mathrm{air}}$ are the reference path gain at unit distance and the path-loss exponent of the aerial (A2A and A2G) links, distinct from the ground pair $(\beta_0^{\mathrm{S2U}},\alpha_p^{\mathrm{S2U}})$ of the S2U link.

UAV positions evolve as
\begin{equation}
\qn(t{+}1) = \qn(t) + \vn(t)\,\delta_t, \quad \forall\, n, t,
\label{eq:kinematics}
\end{equation}
with speed limit $\|\vn(t)\| \leq V_{\max}$ and altitude corridor $H_{\min} \leq [\qn(t{+}1)]_3 \leq H_{\max}$, where $[\cdot]_3$ is the vertical component.

\subsection{Sensor Data Model and Age of Information}

Each sensor in cluster~$k$ monitors a continuous-time OU process:
\begin{equation}
d\xi_k(t) = -\theta_k\,\xi_k(t)\,dt + \sigma_k\,dW_k(t), \quad \theta_k > 0,
\label{eq:OU}
\end{equation}
where $\xi_k(t)$ is the state, $\theta_k$ is the mean-reversion rate, $\sigma_k$ the diffusion coefficient and $W_k(t)$ is a standard Wiener process~\cite{Oksendal2003} independent across clusters. \rev{The OU process models physical quantities that fluctuate around an operating point under random disturbances, such as temperature, pressure or pollutant concentration, and it is the standard source model in remote estimation~\cite{Ornee2021}. Its conditional variance gives the FC's estimation error in closed form as a function of the AoI, so the urgency of a cluster follows directly from the estimation task. Because the OU parameters vary across clusters, the monitored processes have different levels of volatility and hence different update priorities.}

\rev{Let $\Delta_k(t)$ be the age of information (AoI) of cluster~$k$ at time slot~$t$, defined as the number of slots elapsed since its last delivered update.} Information freshness for cluster~$k$ is its AoI $\Delta_k(t)$, reset only by a successful end-to-end delivery: the S2U and A2G transmissions must both succeed in the same slot, captured by the successful-service indicator $\widetilde S_k(t)=\sum_{n=1}^{N}\alpha_{k,n}(t)\,\mathbf{1}[\mathcal{E}_{k,n}^{\mathrm{S2U}}(t)\cap\mathcal{E}_{n}^{\mathrm{A2G}}(t)]$, where $\mathcal{E}_{k,n}^{\mathrm{S2U}}(t)$ is the S2U decode event at UAV~$n$ and $\mathcal{E}_{n}^{\mathrm{A2G}}(t)$ the A2G forwarding event in the same slot, so $\widetilde S_k(t)\in\{0,1\}$ under the one-to-one rule. The AoI then evolves as
\begin{equation}
\Delta_k(t{+}1) =
\begin{cases}
1, & \widetilde S_k(t) = 1,\\
\Delta_k(t) + 1, & \text{otherwise,}
\end{cases}
\label{eq:AoI_update}
\end{equation}
with $\Delta_k(1) = \Delta_k^{(0)}$\rev{, where $\Delta_k^{(0)}$ is the given initial AoI, set to $1$ in all experiments}.

The minimum mean-square error~(MMSE) of predicting $\xi_k(t)$ from the last observation, $\Delta_k(t)$ slots ago, is $\sigma_{\xi,k}^2(\Delta_k(t)) = \bigl(\sigma_k^2/(2\theta_k)\bigr)\bigl(1 - e^{-2\theta_k\,\Delta_k(t)\,\delta_t}\bigr)$~\cite{Ornee2021}, which grows monotonically with $\Delta_k(t)$, saturating at $\sigma_k^2/(2\theta_k)$.

\rev{A successful service of cluster~$k$ resets its next-slot AoI to $\Delta_k(t{+}1)=1$. Conditioned on the state at slot~$t$, the expected next-slot estimation error is therefore}
\begin{align*}
\rev{\mathbb{E}\bigl[\sigma_{\xi,k}^2(\Delta_k(t{+}1))\bigr]} &\rev{{}= \sigma_{\xi,k}^2(\Delta_k(t){+}1)} \\
&\rev{{}\quad - \sum_{n=1}^{N}\alpha_{k,n}(t)\,p_{k,n}^{\mathrm{succ}}(t)\,g_k(t),}
\end{align*}
\rev{where $p_{k,n}^{\mathrm{succ}}(t)$ is the probability that \rev{an attempted service} of cluster~$k$ through UAV~$n$ succeeds and $g_k(t)$ is the \emph{service gain},}
\begin{align}
\rev{g_k(t)} &\rev{{}= \sigma_{\xi,k}^2(\Delta_k(t)+1) - \sigma_{\xi,k}^2(1)} \nonumber \\
&\rev{{}= \frac{\sigma_k^2}{2\theta_k}\Bigl[e^{-2\theta_k\delta_t} - e^{-2\theta_k(\Delta_k(t)+1)\delta_t}\Bigr].}
\label{eq:gain}
\end{align}
\rev{The service gain is \rev{the reduction in the next-slot estimation error produced by} one successful service of cluster~$k$. It is nonnegative, \rev{increases} with $\Delta_k(t)$ and saturates at $(\sigma_k^2/2\theta_k)\,e^{-2\theta_k\delta_t}$\rev{. Consequently, older updates and updates from more volatile sources are more valuable.} The expected per-slot \rev{estimation-error} reduction multiplies this gain by the success probability, $p_{k,n}^{\mathrm{succ}}(t)\,g_k(t)$. A transmission is attempted only when the planning model meets the nominal S2U error target $\varepsilon_{\max}$ and the A2G link to the FC is operated above its decoding threshold. For feasible assignments, the scheduler uses $p_{k,n}^{\mathrm{succ}}(t)\approx1$ as a planning approximation. When selecting assignments, the scheduler thus approximates the expected estimation-error reduction $p_{k,n}^{\mathrm{succ}}(t)g_k(t)$ by $g_k(t)$. The resulting overestimation factor is at most $1/p_{k,n}^{\mathrm{succ}}(t)$, so the approximation is accurate when the actual delivery probability is close to one.}

\begin{lemma}[Service monotonicity]
\label{lem:monotone}
\rev{Fix a cluster~$k$ and let the stage cost $c(\cdot)$ be nondecreasing. \rev{Consider} two policies \rev{under} the same realization of all randomness, so that a service attempt in a given slot has the same outcome under both \rev{policies}. If policy $A$ serves cluster~$k$ in every slot in which policy $B$ serves it, and possibly in more, then $\Delta_k^{A}(t) \le \Delta_k^{B}(t)$ for every $t$, and hence $\sum_{t} c(\Delta_k^{A}(t)) \le \sum_{t} c(\Delta_k^{B}(t))$.}
\end{lemma}
\begin{IEEEproof}
\rev{The proof is by induction on $t$, starting from $\Delta_k^{A}(1)=\Delta_k^{B}(1)$. Suppose $\Delta_k^{A}(t)\le\Delta_k^{B}(t)$. If both policies serve cluster~$k$ and the coupled outcome is a success, both AoI values reset to one. If only $A$ serves and the outcome is a success, then $\Delta_k^{A}(t{+}1)=1\le\Delta_k^{B}(t{+}1)$. In every other case both AoI values increase by one, which preserves the ordering. Because $c$ is nondecreasing, the ordering of the AoI values implies the ordering of the accumulated costs.}
\end{IEEEproof}
\rev{Lemma 1 motivates the use of the instantaneous MMSE as a per-slot cost: postponing a feasible service cannot reduce the accumulated cost of a cluster when that cost is nondecreasing in AoI.}

To track long-term service shortfall, we define the per-cluster deficit queue $Z_k(t) \geq 0$ with $Z_k(1) = 0$:
\begin{equation}
Z_k(t{+}1) = \max\!\bigl(Z_k(t) + \omega_k - \widetilde{S}_k(t),\, 0\bigr),
\label{eq:deficit}
\end{equation}
where $\omega_k \in (0,1]$ is a predetermined long-run successful-update rate and $\widetilde S_k(t)\in\{0,1\}$ indicates a successful update in slot~$t$. The \rev{gain $g_k(t)$} is purely instantaneous and carries no memory across slots, so scheduling on it alone can repeatedly pass over clusters of only moderate momentary urgency, letting their long-run service rate fall short. \rev{The deficit queue supplies this memory by tracking the cumulative gap between the long-term target and achieved service rates: it increases by $\omega_k$ in every slot and decreases by one after a successful update~\cite{Neely2010}. A larger deficit raises the cluster's priority in the assignment, so clusters that fall behind are served sooner.}

When cluster~$k$ is served, reconstructing $\xi_k(t)$ within distortion $D_k^{\max}$ (noise-free observation, so distortion arises only from lossy compression) requires, by the Gaussian rate--distortion bound~\cite{Cover2006}, $B_k^{\min}(t)=\bigl[\tfrac{1}{2}\log_2\bigl(\sigma_{\xi,k}^2(\Delta_k(t))/D_k^{\max}\bigr)\bigr]_+$, where $[x]_+=\max\{x,0\}$. We require $d_{k,n}(t)\geq \max\{B_k^{\min}(t),B_0\}$ whenever cluster~$k$ is assigned, with $B_0>0$ a minimum-payload floor preventing zero-bit service from resetting the AoI.

\subsection{Channel Estimation and Achievable Rates}

At the \rev{beginning} of each slot, every CH transmits a short uplink pilot over the shared S2U sub-band, from which UAV~$n$ estimates the channels $\hat h_{k,n}(t)$ for all \rev{cluster--UAV} pairs $(k,n)$. \rev{S2U transmissions are organized into short packets. Each packet consists of $n_{\mathrm{p}}$ symbols and lasts $\tau_{\mathrm{p}}=n_{\mathrm{p}}/B^{\mathrm{s}}$ seconds, where $B^{\mathrm{s}}$ is the S2U bandwidth. The sensing window of duration $\tau_{\mathrm{s},n}(t)$ supports the transmission of \revii{$L^{\mathrm{s}}_{k,n}(t)=B^{\mathrm{s}}\tau_{\mathrm{s},n}(t)$ symbols} as a sequence of such packets, so a window carries at most $L^{\mathrm{s}}_{k,n}(t)/n_{\mathrm{p}}$ packets. Every packet is preceded by its own short pilot, so the CSI used to decode any packet is aged by at most $\tau_{\mathrm{p}}$.} \rev{Pilot-estimation error is treated separately from channel-aging uncertainty, as detailed below,} and pilot overhead is assumed negligible compared with sensing and forwarding.

\rev{Between the transmission of a packet's pilot and its payload,} UAV motion at velocity $\vn(t)$ decorrelates the channel from its estimate via the Jakes temporal correlation~\cite{Jakes1974}
\begin{equation}
\rho_{k,n}(t) = J_0\!\left(\frac{2\pi\|\vn(t)\| f_c \rev{\tau_{\mathrm{p}}}}{\czo}\right),
\label{eq:jakes}
\end{equation}
with $J_0(\cdot)$ the zeroth-order Bessel function of the first kind, $f_c$ the carrier frequency\rev{, $\tau_{\mathrm{p}}$ the packet duration} and $c_0$ the speed of light. \rev{Although UAV motion causes Doppler-induced channel aging, the per-packet pilots limit this aging to at most one packet duration. At the operating point considered, Jakes temporal correlation exceeds $0.98$ even at the maximum UAV speed. The pilot therefore remains a reliable basis for decoding the packet it precedes.}

\rev{Let $h^{\mathrm{p}}_{k,n}(t)$ denote the true channel at the pilot instant and let $\eta^{\mathrm{p}}_{k,n}(t)$ denote the pilot-estimation error, so that $\hat h_{k,n}(t)=h^{\mathrm{p}}_{k,n}(t)+\eta^{\mathrm{p}}_{k,n}(t)$. Under the Gauss--Markov channel-aging model, the channel at the sensing instant is}
\begin{equation*}
\rev{\begin{aligned}
h_{k,n}(t)
&=\rho_{k,n}(t)h^{\mathrm{p}}_{k,n}(t)+e_{k,n}(t)\\
&=\rho_{k,n}(t)\hat h_{k,n}(t)+e_{k,n}(t)
-\rho_{k,n}(t)\eta^{\mathrm{p}}_{k,n}(t).
\end{aligned}}
\end{equation*}
\rev{Here $\rho_{k,n}(t)$ is the Jakes correlation, and $e_{k,n}(t)$ is a zero-mean proper complex Gaussian innovation independent of both the true pilot-time channel and the pilot-estimation error. With the large-scale channel power held constant over the pilot-to-payload interval, the correlated component has mean power $\rho_{k,n}^2(t)\bar\beta_{k,n}(t)$. Preserving the total mean power $\bar\beta_{k,n}(t)$ therefore gives $\mathbb{E}[|e_{k,n}(t)|^2] =(1-\rho_{k,n}^2(t))\bar\beta_{k,n}(t)$.
This aging variance defines $b_{k,n}(t)$ in the effective SINR, while the separate pilot-error contribution $-\rho_{k,n}(t)\eta^{\mathrm{p}}_{k,n}(t)$ is addressed through the amplitude margins described below. The mean power of the $k \to n$ channel is}
\begin{equation*}
\rev{\bar\beta_{k,n}(t) = \beta_0^{\mathrm{S2U}}\,\|\qn(t{+}1)-\mathbf{q}_k\|^{-\alpha_p^{\mathrm{S2U}}}.}
\end{equation*}
\rev{To compute the effective SINR, we model the innovation residual and the interference uncertainty jointly as Gaussian noise whose variance equals the sum of their variances, assuming these contributions are mutually uncorrelated. The Gaussian model is exact for the innovation under Rayleigh fading. For the interference, it is the pessimistic choice: among all noise distributions with a given variance, Gaussian noise makes decoding most difficult~\cite{Durisi2016}. This worst-case interpretation concerns asymptotic capacity for independent additive noise; at finite blocklength, the Gaussian model is a planning approximation used with conservative amplitude margins. This gives the pessimistic effective SINR}
\begin{equation}
\gamma_{k,n}^{\mathrm{est}}(t)
=
\frac{
\rev{a_{k,n}^2(t)} P_k(t)
}{
\sigma_n^2
+ \rev{b_{k,n}(t)\,P_k(t)}
+ \sum_{k'\neq k} S_{k'}(t)P_{k'}(t)\,\rev{c_{k',n}(t)}
},
\label{eq:est_SINR}
\end{equation}
\rev{where $a_{k,n}(t)$ is the desired-link amplitude after the safety margin, $b_{k,n}(t)$ is the desired-link innovation power, and $c_{k',n}(t)$ is the interference coefficient after the amplitude margin. Let $\sigma_{\eta,k,n}(t)$ denote the standard deviation of either real or imaginary component of the pilot error $\eta^{\mathrm{p}}_{k,n}(t)$. Omitting the slot index in the following definitions,}
\begin{align*}
\rev{a_{k,n}} &\rev{{}= \bigl[|\rho_{k,n}\hat h_{k,n}|-z_{\mathrm{own}}|\rho_{k,n}|\sigma_{\eta,k,n}\bigr]_+,} \\
\rev{b_{k,n}} &\rev{{}= (1-\rho_{k,n}^2)\bar\beta_{k,n},} \\
\rev{c_{k',n}} &\rev{{}= \Bigl(|\rho_{k',n}\hat h_{k',n}|} \\
&\rev{{}\quad+z_{\mathrm{int}}\sqrt{\rho_{k',n}^2\sigma_{\eta,k',n}^{2}+\tfrac12 b_{k',n}}\Bigr)^2.}
\end{align*}
\rev{Both coefficients are known scalars in the current slot, computed from the pilot estimates. The desired-link amplitude is reduced, with a floor of zero, by $z_{\mathrm{own}}$ standard deviations of the aged pilot-error term $\rho_{k,n}(t)\eta^{\mathrm{p}}_{k,n}(t)$. Each interference amplitude is increased by $z_{\mathrm{int}}$ standard deviations of the combined residual $e_{k',n}(t)-\rho_{k',n}(t)\eta^{\mathrm{p}}_{k',n}(t)$, then squared to obtain its interference coefficient. Here each standard deviation refers to either real or imaginary component of the corresponding complex error. For a proper complex Gaussian error, the unconditional probability that its magnitude exceeds $z$ component standard deviations is $e^{-z^2/2}$. Thus, a desired marginal probability $\delta_{\mathrm{m}}$ corresponds to $z=\sqrt{2\ln(1/\delta_{\mathrm{m}})}$. Larger margins make planning more conservative and can reduce the feasible service opportunities. A transmission is \emph{certified} when its promised data volume satisfies the rate bound below under these adjusted estimates.}
Under the Polyanskiy--Poor--Verd\'{u}~(PPV) normal approximation~\cite{Polyanskiy2010}, \rev{at most $n_{\mathrm{p}}C(\gamma_{k,n}^{\mathrm{est}}(t))-\Phi\sqrt{n_{\mathrm{p}}V(\gamma_{k,n}^{\mathrm{est}}(t))}$ bits can be decoded from each packet at block-error probability $\varepsilon_{\mathrm{p}}$; summing over the $L^{\mathrm{s}}_{k,n}(t)/n_{\mathrm{p}}$ packets of the window, the decodable data volume is}
\begin{equation}
d_{k,n}(t) \leq L^{\mathrm{s}}_{k,n}(t)\,C\bigl(\gamma_{k,n}^{\mathrm{est}}(t)\bigr) - \Phi\,\rev{L^{\mathrm{s}}_{k,n}(t)\sqrt{V\bigl(\gamma_{k,n}^{\mathrm{est}}(t)\bigr)/n_{\mathrm{p}}}},
\label{eq:FBL}
\end{equation}
with $C(\gamma) = \log_2(1{+}\gamma)$, $V(\gamma) = 1 - (1{+}\gamma)^{-2}$ the channel dispersion~\cite{Zhu2023} and \rev{$\Phi = Q^{-1}(\varepsilon_{\mathrm{p}})/\ln 2$, the Gaussian tail factor expressed in bits. Here $\varepsilon_{\mathrm{p}} = \varepsilon_{\max}n_{\mathrm{p}}/(B^{\mathrm{s}}\delta_t)$ is the per-packet error budget. A window contains at most $B^{\mathrm{s}}\delta_t/n_{\mathrm{p}}$ packets, so this allocation makes the sum of the nominal packet-error budgets at most $\varepsilon_{\max}$}. \rev{The rate expression~\eqref{eq:FBL} is an approximation rather than a bound, for two reasons. First, the normal approximation becomes exact only as the packet grows long, and our packets are short. Second, the CSI mismatch is treated as Gaussian noise, which as noted above is the pessimistic choice at a given variance. The scheduler makes its decision at the start of the slot, before the per-packet pilots of that slot are observed. It therefore substitutes the slot-start pilot estimate for the per-packet channel estimates. The uncertainty introduced by this substitution, arising from both the innovation and the interference, is incorporated into~\eqref{eq:est_SINR} through the corresponding conditional variances.}

By contrast, the A2G blocklength $B^{\mathrm{A2G}}\tau_{\mathrm{f},n}(t)$ ($B^{\mathrm{A2G}}$ the A2G bandwidth) is large enough that finite-blocklength corrections are negligible, so the A2G rate (operating threshold $\gamma^{\mathrm{A2G}}$) follows Shannon capacity:
\begin{equation}
R_n^{\mathrm{A2G}}(t)
=
B^{\mathrm{A2G}}
\log_2\!\bigl(1+\gamma_n^{\mathrm{A2G}}(t)\bigr).
\label{eq:A2G_rate}
\end{equation}

\subsection{Swarm Connectivity Model}
\label{subsec:aerial_connectivity}

\rev{Swarm connectivity depends on the active A2A links, which succeed only probabilistically under fading. The resulting A2A network is therefore modeled as a random graph. The algebraic connectivity $\bar\lambda_2$ and the Fiedler vector $\vbar_2$ are obtained from its expected Laplacian. The scheduler keeps $\bar\lambda_2$ above a floor and rewards it in the objective through the connectivity price $\mu$. The components of $\vbar_2$ quantify the importance of individual inter-UAV links to overall connectivity and are subsequently used by the matching stage to evaluate the connectivity cost of moving a UAV toward a cluster.}

The A2A SNR $\gamma_{nm}^{\mathrm{A2A}}(t)$ in~\eqref{eq:A2A_A2G_SNR} is the average (path-loss) SNR; under Rician fading with factor $\kappa^{\mathrm{A2A}}$~\cite{Hua2025A2A}, the instantaneous SNR is $\gamma_{nm}^{\mathrm{A2A}}(t)|g_{nm}(t)|^2$ with unit-mean fading gain $|g_{nm}(t)|^2$. Because the A2A links carry low-rate coordination traffic rather than short-packet sensor data, they operate in the long-blocklength regime: finite-blocklength corrections are negligible, as for the A2G link in~\eqref{eq:A2G_rate}, and a link is declared successful when its instantaneous SNR exceeds an operating threshold $\gamma^{\mathrm{A2A}}$, giving the link success (non-outage) probability $\tilde{p}_{nm}(t)=\Pr[|g_{nm}(t)|^2\geq\gamma^{\mathrm{A2A}}/\gamma_{nm}^{\mathrm{A2A}}(t)]$. Under Rician fading the complementary cumulative distribution function~(CCDF) of $|g_{nm}|^2$ is the first-order Marcum~$Q_1$~\cite{Marcum1960}, yielding
\begin{equation}
\tilde{p}_{nm}(t) =
Q_1\!\left(
\sqrt{2\kappa^{\mathrm{A2A}}},
\sqrt{
\frac{2(\kappa^{\mathrm{A2A}}+1)\gamma^{\mathrm{A2A}}}
{\gamma_{nm}^{\mathrm{A2A}}(t)}
}
\right).
\label{eq:marcum}
\end{equation}

The $N$ UAVs are the vertices of a random graph; a candidate edge is a UAV pair $(n,m)$ allowed to form an A2A link, collected in the distance-limited candidate edge set $\mathcal{E}(t)=\{(n,m):n<m,\ \|\qn(t)-\mathbf{q}_m(t)\|\le R_{\mathrm{c}}\}$ for an A2A candidate radius $R_{\mathrm{c}}$, evaluated at the start-of-slot positions so that $\mathcal{E}(t)$ is part of the observed state; $R_{\mathrm{c}}=\infty$ recovers the complete candidate set. Let $a_{nm}(t)\sim\mathrm{Bernoulli}(\tilde p_{nm}(t))$ be the independent indicator that candidate edge $(n,m)$ is present in slot~$t$ (with $a_{nm}(t)\equiv 0$ for $(n,m)\notin\mathcal{E}$), and let $\mathbf{E}_{nm}=(\mathbf{e}_n-\mathbf{e}_m)(\mathbf{e}_n-\mathbf{e}_m)^{\top}$ denote the edge Laplacian, where $\mathbf{e}_n$ is the $n$-th standard basis vector. The \emph{realized} A2A graph then has Laplacian $\mathbf{L}(t)=\sum_{(n,m)\in\mathcal{E}(t)}a_{nm}(t)\,\mathbf{E}_{nm}$, a random matrix. Averaging over the edge successes gives the \emph{expected} graph Laplacian $\Lbar(t)=\mathbb{E}[\mathbf{L}(t)]\in\mathbb{R}^{N\times N}$, the degree matrix minus the probability-weighted adjacency, with entries $[\Lbar(t)]_{nm}=-\tilde{p}_{nm}(t)$ for $n\neq m$ and $[\Lbar(t)]_{nn}=\sum_{j\neq n}\tilde{p}_{nj}(t)$; it is symmetric positive semidefinite for reciprocal links, with eigenvalues $0=\bar{\lambda}_1(t)\leq\cdots\leq\bar{\lambda}_N(t)$.

The second-smallest eigenvalue $\bar{\lambda}_2(t)=\lambda_2(\Lbar(t))$ is the algebraic connectivity \rev{of the expected graph and is positive if and only if that graph is connected}; its \rev{corresponding} eigenvector $\vbar_2(t)$ is the Fiedler vector. \rev{Each Fiedler-vector entry is associated with one UAV. UAVs that are strongly connected tend to have nearly equal entries, while UAVs separated by a weakly connected part of the graph tend to have substantially different entries. A link joining UAVs with widely separated Fiedler entries therefore spans a connectivity bottleneck, and strengthening that link has a comparatively large effect on the algebraic connectivity. More precisely, if the Fiedler eigenvalue is distinct from all other eigenvalues of $\Lbar(t)$, increasing the reliability $\tilde p_{nm}$ of link $(n,m)$ raises $\lambda_2$ at the rate $\partial\lambda_2/\partial \tilde p_{nm}=([\vbar_2]_n-[\vbar_2]_m)^2$, the squared difference of the two endpoint entries.}

Because A2A links fail at random, connectivity cannot be guaranteed in every realization, so we require it with high probability, a chance constraint enforced by two conditions: the expected algebraic connectivity is kept above a target $\lambda_{\min}>0$, and each candidate link is held above a per-edge reliability floor with slack $\eta\in(0,1)$, so the realized graph stays close to its mean. A union bound over the $|\mathcal{E}(t)|$ candidate edges, each failing with probability at most $\eta$, keeps the violation probability within $\delta_c\in(0,1)$ when $|\mathcal{E}(t)|\,\eta\leq\delta_c$; on degree-limited candidate graphs $|\mathcal{E}(t)|=O(N)$, so the admissible slack decays only linearly in the swarm size rather than quadratically.

\subsection{UAV Propulsion and Energy Model}
\label{subsec:propulsion_energy}

Following the standard rotary-wing power model in~\cite{Zeng2019}, the propulsion power for UAV~$n$ at speed $v=\|\vn(t)\|$ is
\begin{align}
P_{\mathrm{fly}}(v) &= P_0\!\left(1 + \frac{3v^2}{U_{\mathrm{tip}}^2}\right) + \frac{d_0\rho_a s_a A}{2}\,v^3 \nonumber\\
&\quad + P_i\left(\sqrt{1+\frac{v^4}{4v_0^4}} - \frac{v^2}{2v_0^2}\right)^{1/2},
\label{eq:propulsion}
\end{align}
where $P_0, P_i$ are the blade-profile and induced hover powers, $U_{\mathrm{tip}}$ the rotor-blade tip speed, $v_0$ the mean rotor-induced hover velocity, $d_0$ the fuselage drag ratio, $\rho_a$ the air density, $s_a$ the rotor solidity and $A$ the rotor-disc area.

Remaining UAV energy $E_n^{\mathrm{rem}}(t)$ depletes through propulsion and relay (the small fixed S2U receive power is omitted, being dominated by propulsion),
\begin{equation}
E_n^{\mathrm{rem}}(t{+}1)
=
E_n^{\mathrm{rem}}(t)
-P_{\mathrm{fly}}(\|\vn(t)\|)\,\delta_t
-P^{\mathrm{relay}}(t)\,\tau_{\mathrm{f},n}(t),
\label{eq:E_uav_update}
\end{equation}
and cluster~$k$'s aggregate energy $E_k^{\mathrm{s,rem}}(t)$, tracked at cluster level since each S2U transmission draws from a shared reservoir regardless of CH rotation, depletes as
\begin{equation}
E_k^{\mathrm{s,rem}}(t{+}1)
=
E_k^{\mathrm{s,rem}}(t)
-\sum_{n=1}^{N}\alpha_{k,n}(t)\,P_k(t)\,\tau_{\mathrm{s},n}(t).
\label{eq:E_sensor_update}
\end{equation}

\section{Problem Formulation}
\label{sec:formulation}

The finite-horizon joint optimization $\mathsf{(P)}$ \rev{minimizes the expected accumulated estimation error of the monitored sources, offset by a weighted connectivity reward,} over cluster--UAV assignment, motion, time allocation, transmit powers and data volumes, subject to finite-blocklength reliability, connectivity and energy constraints:
\begin{subequations}
\label{eq:P}
\allowdisplaybreaks
\begin{align}
\textsf{(P):}\quad
\min_{\mathbf{x}}\quad
& \rev{\mathbb{E}\Bigl[\sum_{t=1}^{T}\sum_{k=1}^{K}
\sigma_{\xi,k}^2\bigl(\Delta_k(t)\bigr)\Bigr]
\;-\;\mu\sum_{t=1}^{T}\bar\lambda_2(t)}
\label{eq:TCU}\\
\text{s.t.}\quad
& d_{k,n}(t)
\leq
\alpha_{k,n}(t)
\Bigl[
L^{\mathrm{s}}_{k,n}(t)C\!\left(\gamma_{k,n}^{\mathrm{est}}(t)\right)
\nonumber\\
&\qquad
-\Phi\,
\rev{L^{\mathrm{s}}_{k,n}(t)\sqrt{
V\!\left(\gamma_{k,n}^{\mathrm{est}}(t)\right)/n_{\mathrm{p}}}}
\Bigr]_+
\label{eq:C1a}\\
& \max\{B_k^{\min}(t),B_0\}\alpha_{k,n}(t)
\leq d_{k,n}(t)
\label{eq:C1b}\\
& \gamma_{n}^{\mathrm{A2G}}(t)
\geq \gamma^{\mathrm{A2G}}
\label{eq:C3}\\
& \bar\lambda_2(t) \;\geq\; \lambda_{\min},
\label{eq:C4}\\
& \tilde p_{nm}(t) \;\geq\; 1 - \eta, \quad \forall\, n < m,
\label{eq:C4_floor}\\
& \qn(t{+}1)
=
\qn(t)+\vn(t)\delta_t
\label{eq:C5}\\
& \|\vn(t)\|
\leq V_{\max}
\label{eq:C6}\\
& H_{\min}
\leq
[\qn(t{+}1)]_3
\leq
H_{\max}
\label{eq:C7}\\
& \|\qn(t{+}1)-\mathbf{q}_m(t{+}1)\|
\geq d_{\mathrm{safe}}
\label{eq:C8}\\
& \sum_{n=1}^{N}\alpha_{k,n}(t)
\leq 1,
\qquad
\sum_{k=1}^{K}\alpha_{k,n}(t)
\leq 1
\label{eq:C9}\\
& \tau_x^{\min}\,{\textstyle\sum_{k}}\alpha_{k,n}(t)\leq\tau_{x,n}(t),
  \quad x\in\{\mathrm{s},\mathrm{f}\},\notag\\
& \tau_{\mathrm{s},n}(t)+\tau_{\mathrm{f},n}(t)\leq \delta_t
\label{eq:C10}\\
& d_{k,n}(t)
\leq
R_n^{\mathrm{A2G}}(t)\tau_{\mathrm{f},n}(t)
+
M\bigl(1{-}\alpha_{k,n}(t)\bigr)
\label{eq:C11}\\
& \sum_{t=1}^{T} \widetilde S_k(t) \geq \revii{\omega_k T}, \quad \forall k
\label{eq:C12}\\
& P_{\mathrm{fly}}\!\left(\|\vn(t)\|\right)\delta_t
+
P^{\mathrm{relay}}(t)\tau_{\mathrm{f},n}(t)
\leq
E_n^{\mathrm{rem}}(t)
\label{eq:C13}\\
& \sum_{n=1}^{N}
\alpha_{k,n}(t)P_k(t)\tau_{\mathrm{s},n}(t)
\leq
E_k^{\mathrm{s,rem}}(t).
\label{eq:C14}
\end{align}
\end{subequations}

\subsection{Problem Components}

The decision vector $\mathbf{x}$ collects $\{\vn(t), \alpha_{k,n}(t), \tau_{\mathrm{s},n}(t), \tau_{\mathrm{f},n}(t), P_k(t), P^{\mathrm{relay}}(t), d_{k,n}(t)\}$ over all $(k,n,t)\in\mathcal{K}\times\mathcal{N}\times\mathcal{T}$, where $\mathcal{K}=\{1,\ldots,K\}$, $\mathcal{N}=\{1,\ldots,N\}$, $\mathcal{T}=\{1,\ldots,T\}$. \rev{In the objective, the expectation is taken over the random service outcomes, $\sigma_{\xi,k}^2(\cdot)$ is the OU conditional variance and $\mu>0$ is the connectivity price. The first term represents the cumulative conditional estimation variance, which is equivalent, up to a constant scaling factor, to the time-averaged mean-square estimation error over the finite horizon. Thus, the objective directly minimizes the estimation error rather than an information-freshness proxy. The second term rewards the algebraic connectivity of the expected swarm Laplacian $\bar\lambda_2(t)$, so every UAV, assigned or idle, has an incentive to keep the swarm graph well connected.}

The constraints fall into four groups: (i)~communication reliability and connectivity~\eqref{eq:C1a}--\eqref{eq:C4_floor}; (ii)~UAV mobility and assignment~\eqref{eq:C5}--\eqref{eq:C9}; (iii)~per-slot resource allocation and coverage~\eqref{eq:C10}--\eqref{eq:C12}; and (iv)~energy causality~\eqref{eq:C13}--\eqref{eq:C14}. Within group~(i), \eqref{eq:C1a}--\eqref{eq:C1b} bound delivered bits between the FBL rate and a minimum payload. The Shannon expression~\eqref{eq:A2G_rate} sets the A2G throughput used in the delivery bound~\eqref{eq:C11} but places no lower bound on the SNR itself: on its own it would let a UAV drift far from the FC, where a vanishingly small A2G SNR still returns a nonzero rate sufficient for a light payload. Constraint~\eqref{eq:C3} closes this gap by holding the A2G SNR at or above the receiver operating threshold $\gamma^{\mathrm{A2G}}$, so every active UAV remains within reliable forwarding range of the FC rather than merely within information-theoretic reach. Constraints~\eqref{eq:C4}--\eqref{eq:C4_floor} enforce probabilistic connectivity through an LMI on the expected algebraic connectivity together with a per-edge reliability floor. In group~(iii), the big-$M$ term in~\eqref{eq:C11} deactivates the delivery bound when $\alpha_{k,n}(t)=0$ and~\eqref{eq:C12} requires \revii{each cluster to receive at least $\omega_k T$ successful updates over the horizon, i.e., its per-slot target rate $\omega_k$ scaled by the horizon length $T$}. The big-$M$ coefficient is set to $B^{\mathrm{s}}\delta_t\log_2(1+\gamma_{\max})$, where $\gamma_{\max}=P_k^{\max}h_{\max}^2/\sigma_n^2$ is the maximum feasible SINR, with $h_{\max}$ an a priori envelope on $|h_{k,n}(t)|$ derived from the minimum UAV--cluster distance under the S2U path-loss model $\beta_0^{\mathrm{S2U}}d^{-\alpha_p^{\mathrm{S2U}}}$. We assume the initial deployment renders $\mathsf{(P)}$ feasible.

\subsection{Per-Slot Problem}
\label{sec:perslot}

For the estimation-error component of the objective, serving cluster~$k$ through UAV~$n$ in slot~$t$ reduces the expected next-slot estimation error by $p_{k,n}^{\mathrm{succ}}(t)\,g_k(t)$, whereas the remaining error terms are independent of the slot-$t$ decision. Since every feasible transmission is designed to satisfy the prescribed reliability target, the scheduler uses the gain $g_k(t)$ as the service benefit. Incorporating the connectivity reward, the per-slot problem is formulated as 
\begin{equation}
\rev{\textsf{(P$_t$):}\quad
\max_{\mathbf{x}(t)}\;\;
\sum_{k=1}^{K}\sum_{n=1}^{N}\alpha_{k,n}(t)\,g_k(t)
\;+\;\mu\,\bar\lambda_2(t)}
\label{eq:perslot}
\end{equation}
\rev{subject to the slot-$t$ constraints of~\eqref{eq:P}. Two remarks connect the per-slot problem to the finite-horizon formulation~\eqref{eq:P}. First, the proposed scheduler is the one-step-lookahead policy of~\eqref{eq:P}: at each slot it maximizes the expected one-slot decrease of the objective, and by Lemma~\ref{lem:monotone} serving whenever it is feasible never increases the accumulated cost. \revii{However, solving the per-slot problem~\eqref{eq:perslot} is not equivalent to solving~\eqref{eq:P}. A greedy choice that maximizes the per-slot reward lowers only the immediate cost and does not, in general, minimize the accumulated cost of~\eqref{eq:P}, because it ignores the impact of the current action on future states.} Second, one constraint of~\eqref{eq:P} has no slot-$t$ counterpart. The \rev{long-term} coverage constraint~\eqref{eq:C12} requires \rev{each cluster to achieve a prescribed service rate over the entire horizon and cannot be enforced directly within a single slot}. The deficit and starvation mechanisms address it at the slot level by raising the priority of clusters that fall behind.}

\subsection{Sources of Nonconvexity}

\rev{On the connectivity side, $\lambda_2(\Lbar)$ is concave in the edge weights, so the reward $\mu\bar\lambda_2(t)$ admits an exact convex representation; nonconvexity enters only through the Marcum-$Q$ dependence of the edge probabilities on kinematics and relay power.} \rev{With the system state fixed,} the per-slot subproblem $(P_t)$ is a mixed-integer nonlinear program~(MINLP). Beyond the binaries $\alpha_{k,n}(t)$, the continuous nonconvexities are: \eqref{eq:C1a} is non-concave in $(P_k,\tau_{\mathrm{s},n})$; the Marcum-$Q$ probabilities in~\eqref{eq:C4_floor} are non-convex in kinematic and power variables; \eqref{eq:C8} is reverse-convex; \eqref{eq:C14} is trilinear in $\alpha_{k,n}P_k\tau_{\mathrm{s},n}$; \eqref{eq:C11} contains an inverse-power-law $\beta_{n,\mathrm{FC}}(t)$ in velocity and a non-jointly-concave product $R_n^{\mathrm{A2G}}\tau_{\mathrm{f},n}$; and the propulsion model~\eqref{eq:propulsion} is also nonconvex.

\section{Solution Methodology}
\label{sec:solution}

We solve~\eqref{eq:P} through a per-slot myopic decomposition. At each slot~$t$, the observed state $\mathcal{S}(t)=\{\Delta_k(t),\qn(t),Z_k(t),E_n^{\mathrm{rem}}(t),E_k^{\mathrm{s,rem}}(t)\}$ is fixed, and the single-slot problem is handled in three steps: (i)~the binaries $\alpha_{k,n}$ are relaxed to $[0,1]$ and nonconvex terms replaced by convex surrogates via SCA; (ii)~the relaxed convex subproblem is solved, with the horizon coverage constraint~\eqref{eq:C12} approximated by the deficit and starvation mechanisms; (iii)~integer assignments are recovered by Hungarian matching, followed by feasibility repair and state updates. 


\subsection{Exact Representation of the Per-Slot Objective}
\label{sec:obj_surrogate}

\rev{The per-slot objective~\eqref{eq:perslot} consists of an affine service-urgency term and the connectivity reward $\mu\bar\lambda_2(t)$. Let $\mathbf{w}$ denote the auxiliary edge weights, one weight per candidate inter-UAV link, and define the corresponding weighted Laplacian as $\Lbar(\mathbf{w})$. Because $\Lbar(\mathbf{w})$ is affine in $\mathbf{w}$, its algebraic connectivity is concave in the edge weights: $\lambda_2(\Lbar(\mathbf{w})) = \min_{\mathbf{u}\perp\mathbf{1},\,\|\mathbf{u}\|=1}\mathbf{u}^\top\Lbar(\mathbf{w})\,\mathbf{u}$ is a minimum of functions that are linear in $\mathbf{w}$, and a minimum of linear functions is concave. Therefore, the connectivity reward can be maximized exactly through an epigraph variable. Let $\mathbf{P}_{\perp} = \mathbf{I} - \frac{1}{N}\mathbf{1}\mathbf{1}^{\top}$ be the orthogonal projector onto the complement of $\mathbf{1}$ and let $s$ be a scalar variable. The subproblem objective at slot~$t$ is}
\begin{equation}
\rev{\sum_{k=1}^{K}\sum_{n=1}^{N}
\alpha_{k,n}\bigl[g_k(t)+\lambda_Z Z_k(t)\bigr]
\;+\;\mu\, s,}
\label{eq:augmented_obj}
\end{equation}
\rev{maximized subject to}
\begin{equation}
\rev{\Lbar(\mathbf{w}) - s\,\mathbf{P}_{\perp} \succeq \mathbf{0},
\qquad s \ge \lambda_{\min},}
\label{eq:C4_LMI}
\end{equation}
where $g_k(t)$ is the deterministic service gain~\eqref{eq:gain} and $\lambda_Z>0$ weights the deficit queue~\eqref{eq:deficit}, included so that the continuous solve and the later matching use the same priorities.

The LMI in~\eqref{eq:C4_LMI} is equivalent to $\lambda_2(\Lbar(\mathbf{w})) \ge s$. To see this, note that $\Lbar(\mathbf{w})$ has $\mathbf{1}$ as its zero eigenvector and $\mathbf{P}_{\perp}\mathbf{1}=\mathbf{0}$, so along $\mathbf{1}$ the LMI holds with equality. For every $\mathbf{u}\perp\mathbf{1}$ it requires $\mathbf{u}^{\top}\Lbar(\mathbf{w})\,\mathbf{u} \geq s\|\mathbf{u}\|^2$, and the smallest value of $\mathbf{u}^{\top}\Lbar(\mathbf{w})\,\mathbf{u}/\|\mathbf{u}\|^2$ over $\mathbf{u}\perp\mathbf{1}$ is $\lambda_2(\Lbar(\mathbf{w}))$. Because $s$ enters the objective with the positive weight $\mu$, the optimizer pushes $s$ up until $s=\lambda_2(\Lbar(\mathbf{w}))$, so the reward is exact. The constraint $s\ge\lambda_{\min}$ then enforces the connectivity floor~\eqref{eq:C4} directly. Both objective terms are affine in $(\boldsymbol{\alpha},s)$. The auxiliary edge weights $\mathbf{w}$ and the linearization of the Marcum-$Q$ dependence of the cap $\tilde p_{nm}$ on positions and relay power are constructed below.

\subsection{Convex Surrogates for Nonconvex Constraints}
\label{sec:surrogates}

We describe convex surrogates for nonconvex communication and mobility constraints~\eqref{eq:C1a}, \eqref{eq:C4}--\eqref{eq:C4_floor}, \eqref{eq:C8}, and~\eqref{eq:C11}.

\paragraph{\eqref{eq:C1a}: FBL Data Capacity}
\rev{The finite-blocklength expression~\eqref{eq:C1a} contains a dispersion penalty in which the SINR appears under a square root. We first remove this dependence using the uniform bound $V(\gamma)\leq 1$. The penalty is then at most $(\Phi/\sqrt{n_{\mathrm{p}}})\,n_{k,n}$, which is linear in the window duration $\tau_{\mathrm{s},n}$. Replacing the original penalty by this upper bound yields a conservative capacity constraint. For $n_{\mathrm{p}}=10^3$ the tightening costs at most $\Phi/\sqrt{n_{\mathrm{p}}}\approx 0.26$ bits per symbol. After this replacement,} \eqref{eq:C1a} has \rev{two} nonconvexities, handled in \rev{two} layers: Layer~1 lower-bounds the rate via the quadratic transform \rev{and} Layer~2 \rev{handles} the blocklength--rate product.

\emph{Layer~1 (Quadratic transform).}
The right-hand side of~\eqref{eq:C1a} is non-concave in $(P_k,\tau_{\mathrm{s},n},P_{k'})$ because the estimated SINR is a ratio of decision variables. Rewrite~\eqref{eq:est_SINR} as
$\gamma_{k,n}^{\mathrm{est}}=a_{k,n}^{2}P_k/D_{k,n}$,
where the desired-link amplitude $a_{k,n}$ (the link from cluster~$k$ to UAV~$n$) is a known scalar, whereas the interference-plus-noise denominator $D_{k,n}$ is not constant: it is built from the other clusters' transmit powers and service decisions \rev{and, through the residual term $b_{k,n}P_k$, from cluster~$k$'s own power}.

This denominator is bilinear in $(\alpha_{k',n}, P_{k'})$ through the service indicators $S_{k'}$ in~\eqref{eq:est_SINR}. We expose this as an auxiliary variable $\pi_{k'} = S_{k'}P_{k'}$ for all $k'$, bounded by its McCormick envelope~\cite{McCormick1976} over the box $S_{k'}\in[0,1]$, $P_{k'}\in[0,P_{k'}^{\max}]$. \rev{Collecting the auxiliaries in $\boldsymbol{\pi}=\{\pi_{k'}\}$, the interference-plus-noise denominator becomes affine in $(P_k,\boldsymbol{\pi})$,}
\begin{equation*}
\rev{D_{k,n}=\sigma_n^2+b_{k,n}\,P_k+\sum_{k'\neq k}c_{k',n}\,\pi_{k'},}
\end{equation*}
\rev{where $b_{k,n}$ and $c_{k',n}$ are the known coefficients of~\eqref{eq:est_SINR}.}

The remaining non-concavity is $\gamma_{k,n}^{\mathrm{est}} = a_{k,n}^2 P_k/D_{k,n}$ inside $\log_2(1+\gamma_{k,n}^{\mathrm{est}})$. The quadratic transform~\cite{Shen2018} lower-bounds a fractional concave function around an iterate. With $a_{k,n}^{(i)}, P_k^{(i)}, D_{k,n}^{(i)}$ the values of $a_{k,n}, P_k, D_{k,n}$ at the current SCA iterate~$i$, defining $y_{k,n}^{(i)} = a_{k,n}^{(i)}\sqrt{P_k^{(i)}}/D_{k,n}^{(i)}$ yields the concave lower bound
\begin{equation}
C(\gamma_{k,n}^{\mathrm{est}})
\geq
\log_2\!\left(
1
+
2y_{k,n}^{(i)}a_{k,n}^{(i)}\zeta_{k,n}
-
(y_{k,n}^{(i)})^2 D_{k,n}
\right),
\label{eq:QT}
\end{equation}
where $\zeta_{k,n}\geq 0$ is a concave surrogate for $\sqrt{P_k}$ enforced by the second-order-cone~(SOC) constraint $\zeta_{k,n}^2\leq P_k$. The log of an affine argument in~\eqref{eq:QT} is represented by the exponential-cone constraint $r_{k,n} \leq \log_2(1+u_{k,n})$, with $r_{k,n}$ the auxiliary rate variable lower-bounding $C(\gamma_{k,n}^{\mathrm{est}})$ and $u_{k,n}$ the affine argument $2y_{k,n}^{(i)}a_{k,n}^{(i)}\zeta_{k,n} - (y_{k,n}^{(i)})^2 D_{k,n}$.

\emph{Layer~2 (Handling \(n_{k,n}r_{k,n}\)).}
\rev{The remaining nonconvexity is the product of two decision variables, the blocklength $n_{k,n}$ and the rate $r_{k,n}$. We give the product its own variable, $w_{k,n}=n_{k,n}r_{k,n}$, and replace it by its McCormick envelope, the standard linear relaxation built from lower and upper bounds $[n^L,n^U]\times[r^L,r^U]$ on the two factors. The bounds are retightened at every SCA iteration by optimization-based bound tightening~(OBBT)~\cite{Gleixner2017}, which keeps the envelope close to the true product. The rate cap $r^U$ is set per pair to $\log_2(1+a_{k,n}^2 P_k^{\max}/\sigma_n^2)$, the rate of the pair at full power and zero interference. No feasible rate exceeds this value, and the per-pair cap keeps the envelope tight for low-SINR pairs.}

Combining \rev{both} layers, the \eqref{eq:C1a} surrogate at iteration~$i$ is
\begin{subequations}
\label{eq:C1a_surrogate}
\begin{align}
& d_{k,n}
\leq
\alpha_{k,n}\bigl[w_{k,n}-\rev{(\Phi/\sqrt{n_{\mathrm{p}}})\,n_{k,n}}\bigr]_+, \\
& w_{k,n} \leq n_{k,n}r_{k,n} \quad \text{(McCormick)}, \\
& r_{k,n}\leq \log_2(1+u_{k,n}), \quad \zeta_{k,n}^2\leq P_k, \\
& u_{k,n} = 2y_{k,n}^{(i)}a_{k,n}^{(i)}\zeta_{k,n}-(y_{k,n}^{(i)})^2 D_{k,n}, \\
& D_{k,n} = \sigma_n^2 \rev{+ b_{k,n}P_k + \sum_{k'\neq k}c_{k',n}\,\pi_{k'}}, \\
& \pi_{k'} = S_{k'}P_{k'} \quad \text{(McCormick)},
\end{align}
\end{subequations}
\rev{where $y_{k,n}^{(i)}$ and $a_{k,n}^{(i)}$ are the quadratic-transform iterate values of~\eqref{eq:QT}.} The only term coupling the relaxed assignment $\alpha_{k,n}$ nonconvexly with the other variables is the assignment-conditioned product $\alpha_{k,n}[w_{k,n}-\rev{(\Phi/\sqrt{n_{\mathrm{p}}})\,n_{k,n}}]$; relaxing it through a McCormick envelope over $\alpha_{k,n}\in[0,1]$ and $w_{k,n}\in[w^L,w^U]$ makes the surrogate jointly convex in $(\alpha_{k,n}, P_k, \tau_{\mathrm{s},n}, \boldsymbol{\pi}, \zeta_{k,n}, r_{k,n}, u_{k,n}, w_{k,n}, d_{k,n})$.

\begin{proposition}[FBL surrogate with post-recovery feasibility verification]
\label{prop:fbl_inner}
Fix the SCA linearization point and McCormick-tightened bounds. Any feasible point
$(\alpha_{k,n}, P_k, \tau_{\mathrm{s},n}, \boldsymbol{\pi}, \zeta_{k,n}, r_{k,n}, u_{k,n}, w_{k,n}, d_{k,n})$
of \eqref{eq:C1a_surrogate} is a candidate that, after integer recovery, is verified
against \eqref{eq:C1a}; if \rev{that constraint is} violated, the feasibility-repair routine zeros the most-violating assignment and re-solves.
\end{proposition}
\begin{IEEEproof}
The quadratic transform~\eqref{eq:QT} provides a concave lower bound on
$C(\gamma_{k,n}^{\mathrm{est}})$ that is tight at $y_{k,n}^{(i)}$. The McCormick envelope on $w_{k,n}=n_{k,n}r_{k,n}$ over the
OBBT-tightened box $[n^L,n^U]\times[r^L,r^U]$ is a valid relaxation, tight at the box vertices but not necessarily at interior points. \rev{The uniform bound $V(\gamma)\leq 1$ makes $(\Phi/\sqrt{n_{\mathrm{p}}})\,n_{k,n}$ a valid upper bound on the dispersion penalty of~\eqref{eq:C1a}.} Thus, surrogate feasibility implies feasibility for the modeled PPV surrogate up to the McCormick relaxation gap, which is removed by the post-recovery verification step that evaluates~\eqref{eq:C1a} directly with the recovered integer assignment $\alpha^{*}$ and triggers feasibility repair on violation.
\end{IEEEproof}

Proposition~\ref{prop:fbl_inner} justifies treating~\eqref{eq:C1a_surrogate} as the FBL constraint inside every convex subproblem $Q_t^{(i)}$: the solver may return a point that satisfies the surrogate but not the original PPV bound, so Step~4e of Algorithm~\ref{alg:tcu_sca} re-checks~\eqref{eq:C1a} at the recovered integer assignment and invokes feasibility repair whenever the relaxation gap produces a violation.

\paragraph{\eqref{eq:C4}--\eqref{eq:C4_floor}: Probabilistic Algebraic Connectivity}
These constraints enforce the chance constraint $\Pr[\lambda_2(\mathbf{L})\geq\lambda_{\min}]\geq 1-\delta_c$ on the realized A2A Laplacian $\mathbf{L}$ and its expectation $\Lbar=\mathbb{E}[\mathbf{L}]$, evaluated at the current slot with the index~$t$ suppressed. With confidence parameter $\delta_c\in(0,1)$, we reformulate this requirement in three stages: (i)~replace the random Laplacian by its expected counterpart through auxiliary edge-reliability variables and a per-edge reliability floor; (ii)~impose a deterministic LMI on the expected Laplacian, with the floor and spectral monotonicity yielding the chance constraint via a union bound; (iii)~linearize the Marcum-$Q$ dependence of the per-edge floor on positions and relay power, with post-recovery verification of the original nonlinear condition.

\emph{Step~1 (Auxiliary decoupling).}
Introduce auxiliary edge-reliability variables $w_{nm}$ for $n<m$ with 
$1-\eta \leq w_{nm}\leq \tilde{p}_{nm}$. Since $\lambda_2$ is monotone non-decreasing in the edge weights, this decouples the LMI from the 
Marcum-$Q$ nonlinearity. The lower bound enforces the per-edge reliability floor and is feasible if and only if $\tilde p_{nm}\geq 1-\eta$, which the SCA linearization in Step~3 maintains via its upper bound.

\emph{Step~2 (Linear matrix inequality).}
\rev{The epigraph LMI~\eqref{eq:C4_LMI}, together with its constraint $s \geq \lambda_{\min}$, enforces $\lambda_2(\Lbar(\mathbf{w})) \geq s \geq \lambda_{\min}$.}
The per-edge floor enters as the linear constraint $w_{nm}\geq 1-\eta$ on the auxiliary variables; with $w_{nm}\leq\tilde p_{nm}$ from Step~1, this yields $\tilde p_{nm}\geq 1-\eta$.

\begin{proposition}[Sufficient chance-connectivity condition]
\label{prop:chance_connectivity}
Assume Bernoulli A2A edge outcomes on the candidate set $\mathcal{E}(t)$ with success probabilities
$\{\tilde p_{nm}(t)\}$ each satisfying $\tilde p_{nm}(t) \geq 1-\eta$ with
$\eta \leq \delta_c/|\mathcal{E}(t)|$. If the expected Laplacian satisfies
$\lambda_2(\Lbar(\mathbf{w})) \geq \lambda_{\min}$, then
$\Pr[\lambda_2(\mathbf{L}) \geq \lambda_{\min}] \geq 1-\delta_c$.
\end{proposition}
\begin{IEEEproof}
Let $\mathcal{G}$ be the event that all $|\mathcal{E}(t)|$ candidate edges succeed in slot~$t$. A union bound, which requires no independence between edges, gives $\Pr[\mathcal{G}^c]\leq\sum_{(n,m)\in\mathcal{E}(t)}(1-\tilde p_{nm})\leq|\mathcal{E}(t)|\,\eta\leq\delta_c$, so $\Pr[\mathcal{G}]\geq 1-\delta_c$.

On $\mathcal{G}$ the realized Laplacian is the unweighted Laplacian of the candidate graph, $\mathbf{L}=\sum_{(n,m)\in\mathcal{E}(t)}\mathbf{E}_{nm}$, which dominates the expected Laplacian in the positive-semidefinite order: $\mathbf{L}-\Lbar(\mathbf{w})=\sum_{(n,m)\in\mathcal{E}(t)}(1-w_{nm})\,\mathbf{E}_{nm}\succeq\mathbf{0}$. By Weyl's inequality on the subspace orthogonal to $\mathbf{1}$, this ordering passes to the second-smallest eigenvalue, so the Step-2 LMI yields $\lambda_2(\mathbf{L})\geq\lambda_2(\Lbar(\mathbf{w}))\geq\lambda_{\min}$ on $\mathcal{G}$.

Combining the two stages, $\Pr[\lambda_2(\mathbf{L})\geq\lambda_{\min}]\geq\Pr[\mathcal{G}]\geq 1-\delta_c$.
\end{IEEEproof}

Proposition~\ref{prop:chance_connectivity} is what lets the scheduler enforce the probabilistic connectivity requirement~\eqref{eq:C4}--\eqref{eq:C4_floor} with two deterministic, convex-representable conditions: the expected-Laplacian LMI~\eqref{eq:C4_LMI} \rev{with its epigraph constraint $s\ge\lambda_{\min}$} and the per-edge floor $\tilde p_{nm}\ge 1-\eta$. Whenever the parameters satisfy $\eta\le 2\delta_c/[N(N-1)]$, meeting these two conditions inside each subproblem $Q_t^{(i)}$ certifies $\Pr[\lambda_2(\mathbf{L})\ge\lambda_{\min}]\ge 1-\delta_c$ without sampling the random graph, and the floor is re-verified in Step~4e of Algorithm~\ref{alg:tcu_sca}.

\emph{Step~3 (CW-SCA on Marcum-$Q$).}
Let $\ell_{nm}=\|\boldsymbol{\delta}_{nm}\|$ denote the inter-UAV distance, so $\tilde{p}_{nm}$ in~\eqref{eq:marcum} is a function of $\ell_{nm}$ and $P^{\mathrm{relay}}$. We replace this dependence with its first-order Taylor expansion:
\begin{equation}
w_{nm}
\leq
\tilde{p}_{nm}^{(i)}
+
c_{d,nm}^{(i)}
\bigl(\hat{\ell}_{nm}-\ell_{nm}^{(i)}\bigr)
+
c_{P,nm}^{(i)}
\bigl(P^{\mathrm{relay}}-P^{(i)}\bigr),
\label{eq:CW_SCA}
\end{equation}
where $\tilde{p}_{nm}^{(i)}$, $\ell_{nm}^{(i)}$, $P^{(i)}$ and $\boldsymbol{\delta}_{nm}^{(i)}$ are the iterate values of $\tilde{p}_{nm}$, $\ell_{nm}$, $P^{\mathrm{relay}}$ and the inter-UAV separation; $\hat{\ell}_{nm}=(\boldsymbol{\delta}_{nm}^{(i)})^{\top}\boldsymbol{\delta}_{nm}/\ell_{nm}^{(i)}$ is the affine first-order distance approximation; and $c_{d,nm}^{(i)}<0$, $c_{P,nm}^{(i)}>0$ are the Marcum-$Q$ derivatives with respect to distance and relay power at iterate~$i$.

The final \eqref{eq:C4} surrogate at iteration~$i$ couples Steps~1--3: the expected-Laplacian LMI~\eqref{eq:C4_LMI}, the linearized reliability constraint~\eqref{eq:CW_SCA} and the per-edge floor $w_{nm}\geq 1-\eta$ for all $n<m$, jointly convex in $(\mathbf{w},\rev{s,}\mathbf{v}_n,P^{\mathrm{relay}})$.

\paragraph{\eqref{eq:C8}: Collision avoidance}
\eqref{eq:C8} is reverse-convex because it lower-bounds $\|\boldsymbol{\delta}_{nm}\|$. With $\boldsymbol{\delta}_{nm}=\qn(t+1)-\mathbf{q}_m(t+1)$ and $\boldsymbol{\delta}_{nm}^{(i)}$ the displacement at iterate~$i$, the first-order Taylor expansion of the norm at the iterate is a global affine under-estimator, $\|\boldsymbol{\delta}_{nm}\|\geq(\boldsymbol{\delta}_{nm}^{(i)})^{\top}\boldsymbol{\delta}_{nm}/\|\boldsymbol{\delta}_{nm}^{(i)}\|$, yielding the valid convex surrogate
\begin{equation}
\frac{(\boldsymbol{\delta}_{nm}^{(i)})^\top \boldsymbol{\delta}_{nm}}{\|\boldsymbol{\delta}_{nm}^{(i)}\|}
\geq d_{\mathrm{safe}}.
\label{eq:C8_surrogate}
\end{equation}

\paragraph{\eqref{eq:C11}: A2G delivery capacity}
\eqref{eq:C11} is nonconvex because $\beta_{n,\mathrm{FC}}$ depends on $\qn(t+1)$ through an inverse power law, and the right-hand side contains the bilinear product $R_n^{\mathrm{A2G}}\tau_{\mathrm{f},n}$. We linearize $R_n^{\mathrm{A2G}}$ jointly in $\vn$ and $P^{\mathrm{relay}}$ via chain-rule sensitivities, keeping the coupling with $\tau_{\mathrm{f},n}$ as an explicit affine product. With \(\ell_{n,\mathrm{FC}}=\|\boldsymbol{\delta}_{n,\mathrm{FC}}\|\) and \(\boldsymbol{\delta}_{n,\mathrm{FC}}=\qn(t+1)-\mathbf{q}_{\mathrm{FC}}\), the affine distance approximation at iteration~\(i\) is $\hat{\ell}_{n,\mathrm{FC}}=(\boldsymbol{\delta}_{n,\mathrm{FC}}^{(i)})^{\top}\boldsymbol{\delta}_{n,\mathrm{FC}}/\ell_{n,\mathrm{FC}}^{(i)}$. The chain-rule distance sensitivity, frozen at iterate~$i$, is
\begin{equation}
c_{\beta,\mathbf{v}}^{(i)}
=
\left.\frac{\partial R}{\partial\beta}\frac{\partial\beta}{\partial\ell}\right|^{(i)}
=
\frac{-\alpha_p^{\mathrm{air}}\beta_0^{\mathrm{air}}(\ell^{(i)})^{-\alpha_p^{\mathrm{air}}-1}\,B^{\mathrm{A2G}}P^{(i)}}
{\ln 2\,(\sigma_{\mathrm{FC}}^2+P^{(i)}\beta^{(i)})}<0.
\end{equation}

Here $R\triangleq R_n^{\mathrm{A2G}}$ is the A2G rate~\eqref{eq:A2G_rate}, $\beta\triangleq\beta_{n,\mathrm{FC}}$ the A2G path gain~\eqref{eq:path_gains}, $\ell\triangleq\ell_{n,\mathrm{FC}}$ the UAV--FC distance, $P\triangleq P^{\mathrm{relay}}$ the relay power and the superscript $(i)$ denotes evaluation at the current iterate (so $\ell^{(i)}$, $\beta^{(i)}$ and $P^{(i)}$ are the iterate values of $\ell_{n,\mathrm{FC}}$, $\beta_{n,\mathrm{FC}}$ and $P^{\mathrm{relay}}$); $\beta_0^{\mathrm{air}}$ and $\alpha_p^{\mathrm{air}}$ are the aerial reference path gain and path-loss exponent from~\eqref{eq:path_gains}, $B^{\mathrm{A2G}}$ the A2G bandwidth and $\sigma_{\mathrm{FC}}^2$ the FC noise power.
Combining this with $\hat{\ell}_{n,\mathrm{FC}}$ gives the velocity-dependent first-order correction to $R_n^{\mathrm{A2G}}\tau_{\mathrm{f},n}$:
\begin{equation}
\Delta_{\mathbf{v}}^{(i)}
=
c_{\beta,\mathbf{v}}^{(i)}
\tau_{\mathrm{f},n}^{(i)}
\bigl(
\hat{\ell}_{n,\mathrm{FC}}
-
\ell_{n,\mathrm{FC}}^{(i)}
\bigr).
\label{eq:vel_correction}
\end{equation}
The prefactors $c_{\beta,\mathbf{v}}^{(i)}$, $\tau_{\mathrm{f},n}^{(i)}$ and $\ell_{n,\mathrm{FC}}^{(i)}$ are frozen at iterate~$i$, and $\hat{\ell}_{n,\mathrm{FC}}$ is linear in the displacement $\boldsymbol{\delta}_{n,\mathrm{FC}}$, itself affine in $\vn$ through $\qn(t{+}1)=\qn(t)+\vn\delta_t$; hence $\Delta_{\mathbf{v}}^{(i)}$ is affine in $\vn$. The corrected local surrogate at iterate $i$ for \eqref{eq:C11} is
\begin{align}
d_{k,n}
&\leq
R_n^{(i)}\tau_{\mathrm{f},n}
+
\mu_n^{(i)}
\tau_{\mathrm{f},n}^{(i)}
\bigl(P^{\mathrm{relay}}-P^{(i)}\bigr)
\nonumber\\
&\quad
+
\Delta_{\mathbf{v}}^{(i)}
+
M(1-\alpha_{k,n}),
\label{eq:C11_surrogate}
\end{align}
where $R_n^{(i)} \triangleq R_n^{\mathrm{A2G}}|^{(i)}$ is the A2G rate at the current iterate, $\Delta_{\mathbf{v}}^{(i)}$ the velocity correction~\eqref{eq:vel_correction} and $\mu_n^{(i)}=\partial R_n^{\mathrm{A2G}}/\partial P^{\mathrm{relay}}|^{(i)}=B^{\mathrm{A2G}}\beta^{(i)}/[\ln 2\,(\sigma_{\mathrm{FC}}^2+P^{(i)}\beta^{(i)})]$ the relay-power sensitivity at iterate~$i$. This surrogate is locally tight at the current iterate.

\subsection{Structural Constraints and Additional Convex Reformulations}
\label{sec:exact}

The remaining constraints are retained exactly or handled by standard convex reformulations.

\paragraph{\eqref{eq:C1b} (rate-distortion floor)}
Because $B_k^{\min}$ is constant at slot~$t$, the assignment-scaled lower bound $B_k^{\min}\alpha_{k,n} \leq d_{k,n}$ is linear in $(\alpha_{k,n},d_{k,n})$ and therefore forms a valid LP constraint. When \(\alpha_{k,n}=0\), it reduces to \(0\leq d_{k,n}\), satisfied by the domain constraint.

\paragraph{\eqref{eq:C3} (A2G SNR feasibility)}
Substituting~\eqref{eq:path_gains} into the A2G SNR constraint and rearranging yields $P^{\mathrm{relay}}\geq(\gamma^{\mathrm{A2G}}\sigma_{\mathrm{FC}}^2/\beta_0^{\mathrm{air}})\,\|\qn(t{+}1)-\mathbf{q}_{\mathrm{FC}}\|^{\alpha_p^{\mathrm{air}}}$ for all $n$. \rev{After substitution of}~\eqref{eq:kinematics}, the right-hand side is a convex function of $\vn$. For $\alpha_p^{\mathrm{air}}\geq 1$, the epigraph $\{(\mathbf{x},u):\|\mathbf{x}\|^{\alpha_p^{\mathrm{air}}}\leq u\}$ is a power cone for rational $\alpha_p^{\mathrm{air}}$, supported by modern conic solvers.

\paragraph{\eqref{eq:C5}--\eqref{eq:C7}, \eqref{eq:C9}, \eqref{eq:C10} (kinematics, mobility, assignment, time)}
The kinematic update~\eqref{eq:kinematics} substitutes next-slot positions out of all referencing constraints, so~\eqref{eq:C5} is not retained as a separate constraint. The speed limit is an SOC constraint. After the kinematic substitution, the altitude $[\qn(t{+}1)]_3=[\qn(t)]_3+[\vn]_3\delta_t$ is affine in $\vn$, so the corridor $H_{\min}\leq[\qn(t{+}1)]_3\leq H_{\max}$ reduces to two linear constraints. Constraints~\eqref{eq:C9}--\eqref{eq:C10} are linear in the assignment and time-allocation variables. The assignment-gated window floors are likewise jointly linear in $(\alpha_{k,n},\tau_{\mathrm{s},n},\tau_{\mathrm{f},n})$ since $\sum_k\alpha_{k,n}\le1$, so they are retained exactly. Under the isolated relaxation $\alpha_{k,n}\in[0,1]$, the assignment matrix is totally unimodular~\cite{Hoffman1956}; however, coupling with other constraints can break integrality, requiring integer recovery later.

\paragraph{\eqref{eq:C13}--\eqref{eq:C14} (energy budgets)}
\eqref{eq:C13} contains the nonconvex $P_{\mathrm{fly}}(\|\vn\|)$ and the bilinear $P^{\mathrm{relay}}\tau_{\mathrm{f},n}$. We upper-bound the propulsion term via the descent lemma~\cite{Nesterov2018}:
\begin{align}
P_{\mathrm{fly}}(\|\vn\|)
&\leq
P_{\mathrm{fly}}(\|\vn^{(i)}\|)
+
\nabla_{\vn}P_{\mathrm{fly}}(\|\vn^{(i)}\|)^\top (\vn-\vn^{(i)}) \nonumber\\
&\quad
+\frac{L_{\mathrm{fly}}}{2}\|\vn-\vn^{(i)}\|^2,
\label{eq:descent}
\end{align}
where $L_{\mathrm{fly}}$ is a Lipschitz constant~\cite{Nesterov2018} for $\nabla P_{\mathrm{fly}}$ on $\{\|\vn\|\leq V_{\max}\}$, yielding a rotated-SOC representable upper bound. The bilinear $P^{\mathrm{relay}}\tau_{\mathrm{f},n}$ is handled by McCormick envelopes. The trilinear sensor-energy term in~\eqref{eq:C14} is handled by recursive McCormick: introduce $z_{k,n}=\alpha_{k,n}P_k$, then envelope $z_{k,n}\tau_{\mathrm{s},n}$ with OBBT-tightened bounds.

\subsection{Approximate Terminal Coverage}
\label{sec:C12}

The hard horizon-coverage constraint~\eqref{eq:C12} cannot be guaranteed exactly under per-slot decomposition. We replace it with two mechanisms.

\emph{(i) Deficit queue.}
The per-cluster deficit queue $Z_k(t)$ from~\eqref{eq:deficit} is a soft long-run-coverage proxy: $\frac{1}{T}\sum_{t=1}^{T}\widetilde S_k(t)\geq\omega_k-[Z_k(T{+}1)-Z_k(1)]/T$.

\emph{(ii) Starvation threshold.}
Let $\Delta_k^{\mathrm{thresh}}\in\mathbb{N}$ be the per-cluster maximum tolerated AoI before starvation control activates. When $\Delta_k(t)\geq\Delta_k^{\mathrm{thresh}}$, cluster~$k$ is force-assigned to the nearest feasible UAV during integer recovery; if none is available, it stays in the matching pool with a starvation bonus added to its weight.

\subsection{Proposed Algorithm}
\label{sec:matching}

Each per-slot problem is solved using a multi-start SCA scheme: every start solves a sequence of frozen-point convex subproblems, recovers integer assignments by maximum-weight bipartite matching and repairs residual infeasibility, with the best of $J$ starts retained.

The convex subproblem at iterate~$i$ is $Q_t^{(i)}:\;\max_{\mathbf{x}}\,\eqref{eq:augmented_obj}$ s.t.\ $\widetilde{\mathcal{C}}^{(i)}$, where $\widetilde{\mathcal{C}}^{(i)}$ collects the convex surrogates~\eqref{eq:C1a_surrogate}, \eqref{eq:C4_LMI} with \eqref{eq:CW_SCA} and the per-edge floor, \eqref{eq:C8_surrogate}, \eqref{eq:C11_surrogate}, the propulsion bound~\eqref{eq:descent} and the remaining constraints of $\mathsf{(P)}$ in exact LP, SOC and power-cone form.

After the SCA loop converges, the relaxed assignments are rounded to a feasible one-to-one schedule by maximum-weight bipartite matching\rev{. The weight assigned to pairing cluster~$k$ with UAV~$n$ is}
\begin{equation}
W_{k,n}(t)=\rev{g_k(t)+\lambda_Z Z_k(t)+\mu\,\widehat{\Delta\lambda}_2^{(k,n)}(t)},
\label{eq:assignment_weights}
\end{equation}
\rev{where the first term is the service gain of the cluster, the second accounts for its accumulated service deficit scaled by $\lambda_Z$ and the third estimates the change in algebraic connectivity induced by the pairing. The first two follow directly from the system state. To construct the topology-dependent term, let all starred quantities denote values at the SCA-converged point. For a candidate pairing $(k,n)$, UAV~$n$ is assumed to take a speed-limited one-slot step toward the cluster head:}
\begin{equation*}
\rev{\hat{\mathbf{q}}_{n,k}=\qn(t)+\min\{V_{\max}\delta_t,\|\mathbf{q}_k-\qn(t)\|\}\,\frac{\mathbf{q}_k-\qn(t)}{\|\mathbf{q}_k-\qn(t)\|},}
\end{equation*}
\rev{while the other UAVs stay at their converged end-of-slot positions $\mathbf{q}_m^{*}$. The step changes the distance of every incident link from its converged value $\ell_{nm}^{*}$ to $\hat\ell_{nm,k}=\|\hat{\mathbf{q}}_{n,k}-\mathbf{q}_m^{*}\|$, and the resulting first-order change in connectivity is}
\begin{equation*}
\rev{\widehat{\Delta\lambda}_2^{(k,n)}(t)=\sum_{m\neq n}([\vbar_2^{*}]_n-[\vbar_2^{*}]_m)^2\!c_{d,nm}^{*}\!(\hat\ell_{nm,k}-\ell_{nm}^{*}),}
\end{equation*}
\rev{where $\vbar_2^{*}$ is the Fiedler vector of $\Lbar$, with any unit eigenvector from the eigenspace used if $\lambda_2$ is repeated, and $c_{d,nm}^{*}<0$ is the Marcum-$Q$ distance derivative of~\eqref{eq:CW_SCA}.}

\rev{Serving an urgent cluster through a UAV whose displacement preserves or improves the swarm topology receives a high weight. Conversely, moving a UAV that is critical for connectivity toward a distant cluster is penalized, according to the same connectivity price $\mu$ as in~\eqref{eq:augmented_obj}.} The Hungarian algorithm then selects the assignment $\alpha^{*}$ that maximizes the total weight $\sum_{k,n}\alpha_{k,n}W_{k,n}(t)$ subject to the one-to-one constraints~\eqref{eq:C9}.

Algorithm~\ref{alg:tcu_sca} details the steps: state observation and feature computation (Steps~1--2); the multi-start SCA loop, which freezes the iterate quantities before solving $Q_t^{(i)}$ to tolerance $\varepsilon_{\mathrm{SCA}}$ (Step~3); starvation-aware Hungarian recovery with sequential feasibility repair, which zeros the most-violating assignment and re-solves in at most $N$ passes before reverting to a hover bailout (Steps~4a--4e); and execution with state update (Step~5). \rev{Every point recovered by the loop is verified against the original nonlinear constraints in Step~4e of Algorithm~\ref{alg:tcu_sca} and repaired on violation, so all reported results are feasible for the original problem. The loop runs to the tolerance $\varepsilon_{\mathrm{SCA}}$ under the iteration cap $I_{\max}$.}

\begin{algorithm}[!t]
\small
\caption{TCU-SCA: Per-Slot Scheduling}
\label{alg:tcu_sca}
\DontPrintSemicolon
\LinesNumbered
\SetKwInOut{Input}{Input}
\SetKwInOut{Output}{Output}
\Input{System parameters; initial state $\mathcal{S}(1)$; tuning constants $\varepsilon_{\mathrm{SCA}}, \rev{\mu}, \lambda_Z, J, I_{\max}$}
\Output{Actions $\{\mathbf{v}_n^*,\alpha_{k,n}^*,\tau_{\mathrm{s},n}^*,\tau_{\mathrm{f},n}^*,P_k^*,P^{\mathrm{relay},*}\}$ per slot}
\For{$t = 1, \ldots, T$}{
  \textbf{Step 1:} Observe $\mathcal{S}(t)=\{\Delta_k,\mathbf{q}_n,Z_k,E_n^{\mathrm{rem}},E_k^{\mathrm{s,rem}}\}$; compute \rev{$g_k(t)$}, $B_k^{\min}(t)$, pilot estimates $\hat h_{k,n}$ and Jakes-correlated coefficients for $\gamma_{k,n}^{\mathrm{est}}$\;
  \textbf{Step 2:} Initialize $\mathbf{x}^{(0)}$ (warm start or hover)\;
  \textbf{Step 3:} SCA inner loop\;
  \For{$j = 1, \ldots, J$ \textup{(multi-start)}}{
    \For{$i = 1, \ldots, I_{\max}$ \textup{(SCA inner loop)}}{
      Freeze iterate quantities: channel and geometry terms, the CW-SCA coefficients in~\eqref{eq:CW_SCA}, the FBL coefficients in~\eqref{eq:C1a_surrogate} and the A2G correction~\eqref{eq:vel_correction}\;
      \textbf{Solve} $Q_t^{(i)}$ (LP+SOC+exp-cone+power-cone\rev{+LMI}) $\to \mathbf{x}^{(i+1)}$; \textbf{break} if $\|\mathbf{x}^{(i+1)}-\mathbf{x}^{(i)}\|_\infty\leq\varepsilon_{\mathrm{SCA}}$\;
    }
  }
  Retain best run (highest objective)\;
  \textbf{Step 4a:} Force-assign starved clusters ($\Delta_k\geq\Delta_k^{\mathrm{thresh}}$) to the nearest feasible UAV\;
  \textbf{Step 4b:} Recover integer assignment by Hungarian matching with weights \rev{$W_{k,n}=g_k+\lambda_Z Z_k+\mu\widehat{\Delta\lambda}_2^{(k,n)}$ per~\eqref{eq:assignment_weights}} $\to\alpha_{k,n}^*$\;
  \textbf{Step 4c:} Re-solve the continuous variables with $\alpha^*$ fixed\;
  \textbf{Step 4d:} Feasibility repair: remove the most-violating assignment and re-solve ($\leq N$ passes); revert to hover (bailout) on residual infeasibility\;
  \textbf{Step 4e:} Verify the original nonlinear constraints~\eqref{eq:C1a}, \eqref{eq:C11}, \eqref{eq:C13}--\eqref{eq:C14} and $\tilde p_{nm}\geq 1-\eta$ at the recovered $\mathbf{x}^*$; on any violation invoke Step~4d\;
  \textbf{Step 5:} Execute $\mathbf{x}^*$; update positions via~\eqref{eq:kinematics}, AoI via~\eqref{eq:AoI_update}, energies via~\eqref{eq:E_uav_update}--\eqref{eq:E_sensor_update} and deficits via~\eqref{eq:deficit}\;
}
\end{algorithm}

\rev{Because the Jakes correlations~\eqref{eq:jakes} and the geometry-dependent channel gains are frozen within each convex subproblem and the linearization in~\eqref{eq:CW_SCA} is a local approximation rather than a global one-sided bound, the resulting SCA procedure is a heuristic. Consequently, neither monotonic improvement of the objective nor convergence to a stationary point of the original problem is guaranteed.}

\section{Simulation Results and Complexity Analysis}
\label{sec:results}

We compare the performance of the proposed scheduler against five baselines:
\begin{itemize}
    \item \emph{BCD}: block coordinate descent on~\eqref{eq:P}, alternating between assignment, position, power and time-allocation blocks. It shares our objective and constraints, isolating the value of the SCA surrogate. \rev{Blocks are initialized from the previous slot's solution; at $t=1$, hover and mid-power initialization is used. BCD runs for at most $I_{\max}^{\mathrm{BCD}}=8$ rounds and stops early when the objective improves by less than $\varepsilon_{\mathrm{SCA}}$ over a full cycle.}
    \item \emph{AoI-STO}~\cite{LongTVT2024}: uses per-slot Hungarian assignment with Lyapunov drift-plus-penalty weights $W_{k,n}(t)=(V+X_k(t))(\Delta_k(t)+1)|\rho_{k,n}\hat h_{k,n}|$, where $V>0$ is the drift-penalty tradeoff and the virtual queue $X_k(t+1)=[X_k(t)-a_{\max}]^+ + \Delta_k(t+1)$ enforces the long-run average AoI bound $a_{\max}$. Continuous variables are optimized by BCD under the same constraints as~\eqref{eq:P}. \rev{It replaces the proposed estimation-derived urgency and connectivity reward with a topology-agnostic drift-plus-penalty weight. Paired with BCD, it isolates the contribution of the proposed weighting.}
    \item \emph{AoI-only}: \revii{a greedy maximum-AoI baseline, inspired by classical AoI scheduling}\cite{Kadota2018Scheduling}, \revii{but not identical to it,} \revii{in which} clusters \revii{are} ranked by $\Delta_k(t)$ and matched  \revii{ to UAVs} by Hungarian assignment. It isolates the contribution of spatial planning.
    \item \emph{Random}: each UAV samples one cluster uniformly per slot, with collisions resolved by leaving conflicting UAVs unassigned.
    \item \emph{Nearest}: each UAV greedily flies toward the closest unassigned cluster, ignoring AoI and topology.
\end{itemize}

The heuristic baselines skip the joint continuous optimization, applying max sensor power, max relay power and a $50/50$ sensing/forwarding split. \rev{All methods, heuristics included, are evaluated using the same reliability-certified physical-layer model. Before any transmission, the scheduled data volume is verified against the pessimistic caps of Step~4e, and every method's deliveries are then decoded using identical channel realizations across methods. Any performance difference between methods therefore reflects their scheduling decisions rather than inconsistencies in physical-layer modeling. The connectivity price $\mu$ is calibrated once, on a pilot run, so that the reward contributes $15\%$ of the objective. This gives $\mu=0.8754$, held fixed across all methods and scenarios.} Each method is evaluated over $24$ seeds.

We consider $N=4$ UAVs collecting from $K=15$ sensor clusters over a $1\,\mathrm{km}\times 1\,\mathrm{km}$ area with the FC at the origin, over $T=100$ slots of duration $\delta_t=1$\,s, unless otherwise stated. Cluster dynamics follow the OU model in~\eqref{eq:OU} with parameters drawn independently as $\theta_k\sim\mathcal{U}[0.01,0.20]$, $\sigma_k\sim\mathcal{U}[0.5,2.0]$, $D_k^{\max}\sim\mathcal{U}[0.2,1.0]$. \rev{The draw is performed once: the realized parameter set $\{\theta_k,\sigma_k,D_k^{\max}\}_{k=1}^{K}$ is held fixed across all seeds and all methods, so seed-to-seed variability reflects channel, noise and deployment randomness rather than source heterogeneity, and every method is evaluated on the identical set of sources. The realized set is listed in Table~\ref{tab:ou_params}.} Remaining parameters are listed in Table~\ref{tab:sim_params}. \rev{The operating point $N=4$, $K=15$ represents a small swarm with an average of $K/N=3.75$ sensor clusters per UAV. Since each UAV serves at most one cluster per slot, this setting requires the scheduler to distribute update opportunities across slots. The load-imbalanced cross and dispersed Y deployments vary demand concentration and spatial geometry, respectively, while keeping $N$ and $K$ fixed. The scaling study in Section~\ref{subsec:scaling} varies both $N$ and $K$ to examine the effect of assignment load beyond this operating point.}

\rev{Mobility-induced channel variations are incorporated as follows. Since each packet includes its own pilot sequence, the channel estimate used for decoding is at most one packet duration old. Over this interval, the channel remains highly correlated, with the Jakes temporal correlation exceeding $0.98$ even at the maximum UAV speed considered. Any residual channel variation is accommodated by the safety margins used in the certification procedure.}

\begin{table*}[!t]
\centering
\caption{Simulation parameters.}
\label{tab:sim_params}
\footnotesize
\renewcommand{\arraystretch}{1.0}
\setlength{\tabcolsep}{3pt}
\begin{tabular}{@{}llr@{\hskip 16pt}llr@{}}
\toprule
Symbol & Description & Value & Symbol & Description & Value \\
\midrule
\multicolumn{3}{@{}l}{\textit{Topology and horizon}} & \multicolumn{3}{@{}l}{\textit{OU sensor model (heterogeneous, $\mathcal{U}$)}} \\
$N$ & Number of UAVs & $4$ & $\theta_k$ & Mean-reversion rate & $[0.01,0.20]$ \\
$K$ & Number of sensor clusters & $15$ & $\sigma_k$ & Diffusion coefficient & $[0.5,2.0]$ \\
$T$ & Horizon length (slots) & $100$ & $D_k^{\max}$ & Distortion tolerance & $[0.2,1.0]$ \\
$\delta_t$ & Slot duration & $1$\,s & $\Delta_k^{\mathrm{thresh}}$ & Starvation threshold (slots) & $40$ \\
$\tau_s^{\min},\tau_f^{\min}$ & Min.\ sensing/forwarding window & $0.05\,\delta_t$ & $B_0,\omega_k$ & Min.\ payload / desired service rate & $32$\,bits, \rev{$N/K\!\approx\!0.27$}/slot \\
\midrule
\multicolumn{3}{@{}l}{\textit{Kinematics and geometry}} & \multicolumn{3}{@{}l}{\textit{Connectivity}} \\
$V_{\max}$ & Max.\ UAV speed & \rev{$5$}\,m/s & $\lambda_{\min}$ & Target algebraic connectivity & $0.5$ \\
\rev{$n_{\mathrm{p}}$} & \rev{S2U packet length} & \rev{$10^{3}$\,sym.} & $R_{\mathrm{c}}$ & A2A candidate radius & $\infty$ (complete) \\
$H_{\min},H_{\max}$ & Altitude bounds & $50,150$\,m & $\delta_c$ & Connectivity confidence & $0.3$ \\
$d_{\mathrm{safe}}$ & Inter-UAV separation & $10$\,m & $\eta$ & Per-edge reliability slack & $0.05$ \\
\midrule
\multicolumn{3}{@{}l}{\textit{Channel and RF}} & \multicolumn{3}{@{}l}{\textit{Propulsion (rotary-wing)}} \\
$f_c$ & Carrier frequency & $2.4$\,GHz & $P_0,P_i$ & Blade profile / induced power & $79.86,88.63$\,W \\
$B^{\mathrm{s}}, B^{\mathrm{A2G}}, B^{\mathrm{A2A}}$ & S2U / A2G / A2A bandwidth & $1,10,1$\,MHz & $U_{\mathrm{tip}}$ & Rotor tip speed & $120$\,m/s \\
$\beta_0^{\mathrm{S2U}},\beta_0^{\mathrm{air}}$ & S2U / A2A--A2G path gain at $1$\,m & $-30,-30$\,dB & $v_0$ & Hover induced velocity & $4.03$\,m/s \\
$\alpha_p^{\mathrm{S2U}},\alpha_p^{\mathrm{air}}$ & S2U / A2A--A2G exponent & $2.2,2.2$ & $d_0,s_a$ & Fuselage drag ratio / rotor solidity & $0.6,0.05$ \\
$N_0$ & Noise spectral density & $-174$\,dBm/Hz & $\rho_a,A$ & Air density / rotor disc area & $1.225$\,kg/m$^3$, $0.503$\,m$^2$ \\
$\kappa^{\mathrm{A2A}}$ & Rician $K$-factor (A2A) & $10$ & \multicolumn{3}{@{}l}{\textit{Algorithm parameters}} \\
$\gamma^{\mathrm{A2A}}, \gamma^{\mathrm{A2G}}$ & SNR thresholds (A2A, A2G) & $35,3$\,dB & \rev{$\mu$} & \rev{Connectivity price} & \rev{$0.8754$ ($0.15$ obj.\ share)} \\
$\varepsilon_{\max}$ & FBL target block-error prob. & $10^{-5}$ & $J,I_{\max}$ & SCA warm starts / max iterations & $4,12$ \\
$P_k^{\min},P_k^{\max}$ & Sensor TX power range & $10^{-3},2$\,W & $\varepsilon_{\mathrm{SCA}}$ & SCA convergence tolerance & $10^{-3}$ \\
$P^{\mathrm{relay}}_{\min},P^{\mathrm{relay}}_{\max}$ & UAV relay power range & $10^{-3},1$\,W & \rev{$z_{\mathrm{own}},z_{\mathrm{int}}$} & \rev{Certification safety margins} & \rev{$5,2$} \\
 & & & $\lambda_Z$ & Coverage-pressure weight & $0.5$ \\
\bottomrule
\end{tabular}

\smallskip
\footnotesize Initial conditions: \rev{$\Delta_k(1)=1$ slot,} $E_n^{\mathrm{rem}}(0)=50$\,kJ, $E_k^{\mathrm{s,rem}}(0)=10$\,kJ.
\end{table*}

\begin{table}[!t]
\centering
\footnotesize
\caption{\rev{Fixed OU source parameter set, drawn once and held across all seeds, methods and scenarios.}}
\label{tab:ou_params}
\renewcommand{\arraystretch}{1.0}
\setlength{\tabcolsep}{4.5pt}
\rev{\begin{tabular}{@{}cccc@{\hskip 14pt}cccc@{}}
\toprule
$k$ & $\theta_k$ & $\sigma_k$ & $D_k^{\max}$ & $k$ & $\theta_k$ & $\sigma_k$ & $D_k^{\max}$ \\
\midrule
1 & 0.048 & 1.651 & 0.446 & 9  & 0.195 & 1.308 & 0.316 \\
2 & 0.148 & 1.466 & 0.228 & 10 & 0.090 & 0.542 & 0.812 \\
3 & 0.178 & 1.140 & 0.950 & 11 & 0.173 & 1.422 & 0.913 \\
4 & 0.039 & 0.624 & 0.690 & 12 & 0.196 & 1.437 & 0.978 \\
5 & 0.028 & 1.235 & 0.707 & 13 & 0.100 & 0.795 & 0.974 \\
6 & 0.149 & 0.573 & 0.818 & 14 & 0.116 & 0.665 & 0.241 \\
7 & 0.124 & 1.842 & 0.603 & 15 & 0.022 & 1.532 & 0.462 \\
8 & 0.036 & 0.746 & 0.277 &    &       &       &       \\
\bottomrule
\end{tabular}}
\end{table}

\begin{table*}[!tb]
\centering
\footnotesize
\caption{\rev{Nominal steady-state statistics over slots $50$--$99$ ($24$ seeds, mean $\pm$ std). $\bar\Delta_k$: mean AoI; $\max_k\Delta_k$: steady-state peak AoI; SR: service success rate; MSE: time-averaged OU estimation error $(\sigma_k^2/2\theta_k)(1-e^{-2\theta_k\Delta_k\delta_t})$; Out./Att.: realized outages over service attempts, totaled across seeds. Best in \textbf{bold}.}}
\label{tab:aoi_summary}
\renewcommand{\arraystretch}{1.0}
\setlength{\tabcolsep}{3pt}
\rev{\begin{tabular}{@{}lccccc@{}}
\toprule
Method & $\bar\Delta_k$ & $\max_k\Delta_k$ & SR & MSE & Out./Att. \\
\midrule
Proposed (SCA) & $\mathbf{3.10 \pm 0.25}$ & $12.9$ & $\mathbf{0.983 \pm 0.017}$ & $\mathbf{2.82 \pm 0.22}$ & $\mathbf{0\,/\,9445}$ \\
BCD            & $3.36 \pm 0.52$ & $15.5$ & $0.929 \pm 0.059$ & $2.93 \pm 0.35$ & $0\,/\,8919$ \\
AoI-STO        & $3.73 \pm 0.24$ & $13.8$ & $0.617 \pm 0.032$ & $3.52 \pm 0.16$ & $0\,/\,5925$ \\
AoI-only       & $3.19 \pm 0.16$ & $\mathbf{10.1}$ & $0.732 \pm 0.037$ & $3.19 \pm 0.11$ & $0\,/\,7030$ \\
Nearest        & $55.66 \pm 0.01$ & $100.0$ & $0.879 \pm 0.020$ & $8.78 \pm 1.43$ & $14\,/\,8456$ \\
Random         & $7.47 \pm 0.63$ & $40.9$ & $0.534 \pm 0.029$ & $4.67 \pm 0.37$ & $38\,/\,5166$ \\
\bottomrule
\end{tabular}}
\end{table*}

\subsection{AoI Dynamics and Success Rate}
\label{subsec:aoi_dynamics}

\rev{Fig.~\ref{fig:aoi_baselines} shows the mean and peak OU estimation MSE of all methods over the $T=100$-slot horizon, together with the cumulative success rate, and Table~\ref{tab:aoi_summary} reports steady-state statistics over slots $50$--$99$ together with the estimation-quality and reliability metrics. The proposed scheduler attains the lowest OU estimation MSE ($2.82\pm0.22$, $11.7\%$ below the strongest heuristic) and the highest service success rate ($0.983\pm0.017$) with no observed violation of the prescribed physical-layer reliability condition across $9445$ service attempts. Its mean AoI ($3.10\pm0.25$) is comparable to that of AoI-only ($3.19\pm0.16$). These results show that, when all methods operate over the same reliability-certified physical layer, similar average freshness does not necessarily imply similar estimation accuracy or delivery reliability. The proposed certification and source-aware scheduling mechanisms improve both metrics with negligible AoI penalty.}

\rev{Among the optimization-based baselines, BCD achieves a mean AoI of $3.36\pm0.52$ and a success rate of $0.929\pm0.059$, using $15$ subproblem solves per slot against $48$ for the proposed method. Granting BCD the same $48$-solve budget does not improve its performance: it converges after $37.7$ solves on average and yields a mean AoI of $3.39\pm0.52$ with the same success rate, so the residual gap is attributable to the alternating structure rather than to the computational budget. AoI-STO performs worse in both metrics, with mean AoI $3.73\pm0.24$ and success rate $0.617$: its topology-agnostic drift-plus-penalty policy frequently selects UAV--cluster pairings that are subsequently rejected by the reliability-certification stage. AoI-only matches the proposed method in mean AoI but initiates substantially fewer service attempts, $7030$ versus $9445$. When its selected UAV--cluster pairing fails the certification check, AoI-only provides no alternative assignment, leaving the UAV idle during that slot. By contrast, the proposed scheduler rematches the UAV to another feasible cluster, thereby avoiding the loss of a potential service opportunity.}

\rev{Random and Nearest are the only methods with realized outages, $38$ and $14$ across the $24$ seeds. Both use maximum transmit power and a fixed $50/50$ split and cannot adapt when UAV motion weakens the A2G link to the FC. Under Nearest, each UAV tends to remain close to a small subset of clusters, while the remaining clusters are visited only when their AoI reaches the starvation threshold. Random lacks both an urgency criterion and a directional mobility objective, resulting in a mean AoI of $7.47$.}

\begin{figure*}[!t]
    \centering
    \subfloat[]{%
        \includegraphics[width=0.4\textwidth,height=5cm]{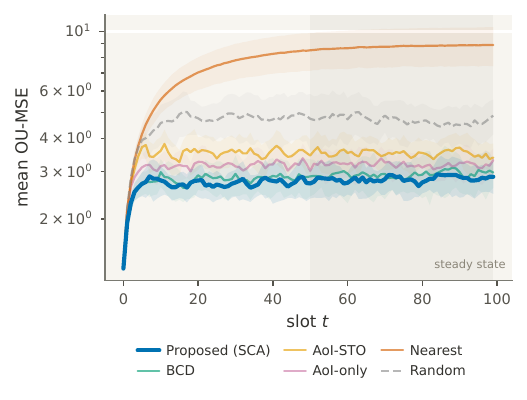}%
        \label{fig:aoi_mean}}
    \hfil
    \subfloat[]{%
        \includegraphics[width=0.4\textwidth,height=5cm]{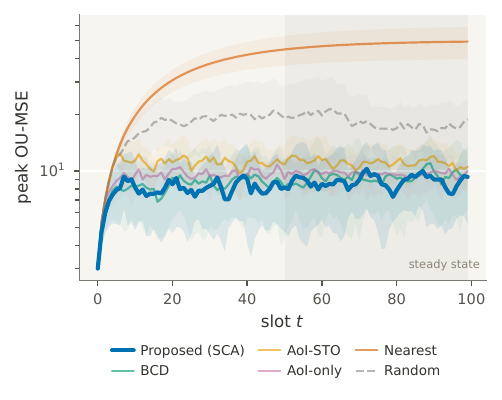}%
        \label{fig:aoi_max}}
    \hfil
    \subfloat[]{%
    \includegraphics[width=0.4\textwidth,height=5cm]{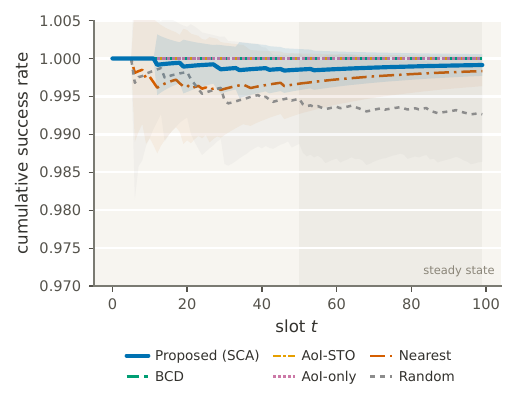}%
        \label{fig:succ_rate}}
    \caption{Baseline comparison over $T=100$ slots, averaged across $24$ seeds: \rev{(a)~mean OU estimation MSE $\tfrac1K\sum_k\sigma_{\xi,k}^2(\Delta_k(t))$; (b)~peak OU estimation MSE $\max_k\sigma_{\xi,k}^2(\Delta_k(t))$;} (c)~cumulative service success rate $\sum_{t'\leq t}\widetilde S_k(t')/\sum_{t'\leq t}S_k(t')$, aggregated across clusters. Shaded bands show $\pm 1$ std.}
    \label{fig:aoi_baselines}
\end{figure*}

\rev{Fig.~\ref{fig:aoi_baselines}(c) shows the success-rate trajectory: the proposed scheduler stabilizes near $0.98$ without transient drops. Certification holds empirically. Across both deployment scenarios and $19{,}040$ proposed-scheduler attempts, and equally for BCD, AoI-STO and AoI-only, not a single delivered update violates the $\varepsilon_{\max}$ target of~\eqref{eq:FBL}. This is the intended effect of the per-packet budget $\varepsilon_{\mathrm{p}}$ and the pilot safety margins: each promised transmission carries enough margin that realized failures stay within the $\varepsilon_{\max}$ budget. In the nominal geometry, the expected algebraic connectivity remains at $\lambda_2=4.00\pm0.00$ throughout the experiment. The candidate graph is complete, and the per-edge reliability constraint~\eqref{eq:C4_floor} is nonbinding, so the connectivity certificate is satisfied at no additional cost. Propulsion energy averages $16.9$\,kJ per UAV over the mission ($33.8\%$ of the onboard budget across the $24$ seeds). Thus, although the energy constraints are enforced throughout all experiments, they remain non-binding at the operating point considered.}

\subsection{Stress Test: Methodology Robustness}
\label{subsec:stress}

\rev{To evaluate robustness to load imbalance and nonuniform spatial geometry, we replace the uniform cluster scatter by a cross-shaped deployment. The clusters are distributed along four radial arms centered at the midpoint of the area, one arm holding three clusters and the other three holding four each, with each cluster placed $300$ to $400$\,m from the center and offset sideways by up to $40$\,m. Every other parameter stays at the Table~\ref{tab:sim_params} operating point: $N=4$, $K=15$, $T=100$, $24$ seeds and all six methods. The difficulty this geometry creates is imbalance. The arms are widely separated and unequally loaded, so a scheduler driven only by the current slot may repeatedly serve clusters on the least costly arms while leaving others unserved for extended periods. The deficit queue is the mechanism that prevents this.}

\begin{table*}[!tb]
\centering
\footnotesize
\caption{\rev{Cross-geometry steady-state statistics over slots $50$--$99$ ($24$ seeds); columns as in Table~\ref{tab:aoi_summary}. Best in \textbf{bold}.}}
\label{tab:stress_summary}
\renewcommand{\arraystretch}{1.0}
\setlength{\tabcolsep}{3pt}
\rev{\begin{tabular}{@{}lccccc@{}}
\toprule
Method & $\bar\Delta_k$ & $\max_k\Delta_k$ & SR & MSE & Out./Att. \\
\midrule
Proposed (SCA) & $2.56 \pm 0.02$ & $7.3$ & $\mathbf{1.000 \pm 0.001}$ & $\mathbf{2.50 \pm 0.01}$ & $\mathbf{0\,/\,9595}$ \\
BCD            & $\mathbf{2.55 \pm 0.01}$ & $7.3$ & $0.998 \pm 0.002$ & $2.50 \pm 0.01$ & $0\,/\,9584$ \\
AoI-STO        & $3.29 \pm 0.15$ & $9.4$ & $0.698 \pm 0.027$ & $3.26 \pm 0.11$ & $0\,/\,6704$ \\
AoI-only       & $2.99 \pm 0.08$ & $8.4$ & $0.777 \pm 0.020$ & $3.06 \pm 0.06$ & $0\,/\,7460$ \\
Nearest        & $55.65 \pm 0.01$ & $100.0$ & $0.890 \pm 0.020$ & $8.84 \pm 1.35$ & $30\,/\,8573$ \\
Random         & $6.11 \pm 0.35$ & $30.7$ & $0.607 \pm 0.019$ & $4.29 \pm 0.22$ & $43\,/\,5874$ \\
\bottomrule
\end{tabular}}
\end{table*}

\rev{Table~\ref{tab:stress_summary} summarizes the steady-state results. The proposed scheduler and BCD both achieve near-ceiling performance, with service success rates of $1.000$ and $0.998$, respectively. Their mean AoI values are nearly identical, at $2.56\pm0.02$ and $2.55\pm0.01$, and their MSE is identical. This close agreement can be attributed to the relatively favorable communication geometry: because the clusters remain close to the area center and the service distances are short, the alternating BCD solver converges to essentially the same operating point as the joint SCA solver. The deficit queue absorbs the heavy-arm imbalance for both methods, with a steady-state peak AoI of $7.3$. The heuristic baselines remain clearly inferior: AoI-STO and AoI-only trail by $0.4$ to $0.7$ slots of mean AoI and $22$ to $30$ success-rate points, and Nearest and Random degrade further, confirming that the cross-shaped deployment stresses load balancing rather than swarm connectivity.}


\subsection{\rev{Connectivity Ablation}}
\label{subsec:ablation}
\rev{In both the uniform and cross-shaped deployments, the candidate graph remains complete and $\lambda_2$ saturates at $N$. Consequently, the connectivity mechanism is enforced without affecting the solution, but its ability to preserve a dispersed topology is not directly tested. To create a connectivity-limited setting, we place the $K$ clusters at the tips of a three-arm Y of arm length $L$ and place the UAVs at the tips as well: at $L=2000$\,m the initial expected connectivity is $\lambda_2=2.67$ and at $L=2400$\,m it is $1.24$, both above the floor $\lambda_{\min}=0.5$ but far from saturation. This geometry creates an explicit trade-off between local service and swarm connectivity. Keeping the UAVs near their respective clusters reduces service distances but weakens the A2A topology, whereas moving the UAVs inward improves connectivity at the cost of increasing their distances from the clusters. The connectivity reward $\mu\bar\lambda_2(t)$ determines how these competing effects are balanced. We run the scheduler twice on the same realization of this geometry, for $T=100$ slots. The first run uses the full objective. The second sets $\mu=0$, which removes the connectivity reward from the objective while keeping the hard connectivity constraint~\eqref{eq:C4} in place.}

\begin{figure}[!t]
    \centering
    \includegraphics[width=0.8\columnwidth]{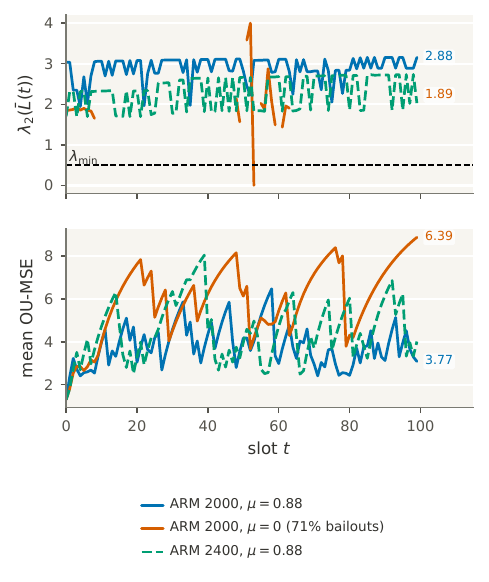}
    \caption{\rev{Fiedler ablation on the dispersed Y geometry, with the UAVs initially at the arm tips. Top: expected connectivity $\lambda_2(\bar L(t))$ with the floor $\lambda_{\min}$ dashed. Bottom: mean OU estimation MSE. Gaps in the $\mu=0$ curves mark bailout slots, in which the solver finds no feasible action and the swarm holds position.}}
    \label{fig:no_fiedler_sweep}
\end{figure}

\rev{Fig.~\ref{fig:no_fiedler_sweep} shows the two runs. With the reward, at $L=2000$\,m the swarm moves inward from the dispersed initial deployment and maintains a mean $\lambda_2$ of $2.88$. Its mean AoI is $5.80$ and its success rate is $0.47$, with no realized outages. The success rate sits below one because distant pairs cannot be certified and are therefore never attempted; every transmission that is actually executed succeeds. At $L=2400$\,m, the connectivity reward increases $\lambda_2$ from its initial value of $1.24$ to a maintained $2.30$, at a mean AoI of $9.57$ and again without a single outage.}

\rev{Removing the reward reverses this behavior. $\lambda_2$ falls below its initial value to a mean of $1.89$; mean AoI rises to $27.82$, corresponding to a $4.8$-fold degradation. The service success rate falls to $0.16$; and one realized outage occurs, the only one in the ablation. Two effects contribute to this degradation. First, without the topology reward, the scheduler no longer steers the swarm away from poorly connected configurations, and the swarm drifts toward the connectivity floor. Second, near the floor, the feasible region defined by the connectivity constraint~\eqref{eq:C4} becomes increasingly restrictive, and $71\%$ of slots end in bailout. The hard constraint keeps every executed action feasible, whereas the connectivity reward maintains sufficient operating margin to avoid repeated bailouts.}

\begin{figure}[!t]
    \centering
    \includegraphics[width=0.8\columnwidth]{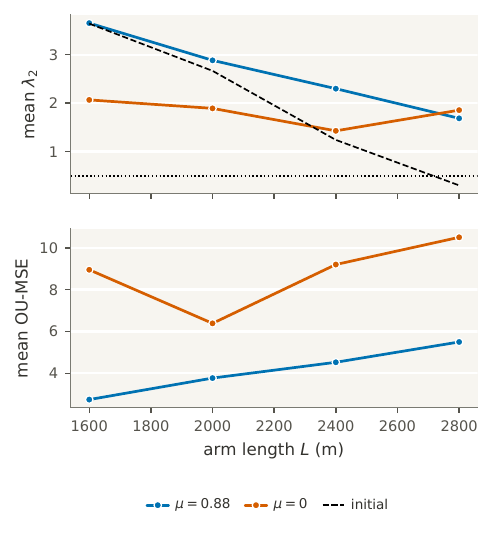}
    \caption{\rev{Arm-length sweep of the Fiedler ablation. Top: mean expected connectivity against its initial value, dashed, and the floor $\lambda_{\min}$, dotted. The $\mu=0$ curve averages $\lambda_2$ over its surviving slots. Bottom: mean OU estimation MSE.}}
    \label{fig:y_sweep}
\end{figure}

\rev{Sweeping the arm length further illustrates this trade-off (Fig.~\ref{fig:y_sweep}). With the reward the swarm holds or restores its margin at every $L\in\{1600,2000,2400,2800\}$\,m: mean $\lambda_2$ stays at $3.65$, $2.88$, $2.30$ and $1.69$ against initial values of $3.64$, $2.67$, $1.24$ and $0.30$, respectively; freshness degrades gracefully from $4.09$ to $15.52$ mean AoI; and no tested arm length produces a single outage or bailout. At $L=2800$\,m the initial deployment itself violates the floor, and the scheduler with the reward recovers the certificate. Without the reward every arm length degrades: mean AoI runs from $28$ to $52$ at success rates of $0.10$ to $0.19$, with $71$ to $86\%$ of slots ending in solver bailout as the deteriorating swarm topology drives the certificate toward infeasibility.}

\subsection{\revii{Effect of the Connectivity Price $\mu$}}
\label{subsec:mu_sweep}
\revii{To characterize the scheduler's sensitivity to the connectivity price, we fix the dispersed $L=1600$\,m Y geometry of Fig.~\ref{fig:y_sweep} and sweep $\mu$ over $[0,5]$ (Fig.~\ref{fig:mu_sweep}). At $\mu=0$, the connectivity reward is absent, and performance deteriorates sharply: $88\%$ of slots end in bailout, the mean connectivity falls to $\bar\lambda_2=1.9$, and the OU estimation MSE rises to $9.67$. All tested positive prices in $\mu\in[0.25,5]$ eliminate bailout and keep the MSE within $16\%$ of the minimum of $2.79$, attained at $\mu=0.8754$. Thus, the $3.5\times$ ratio between the worst and best MSE values is driven primarily by the absence of the connectivity reward, rather than by sensitivity to its positive weight.}

\revii{All tested positive prices outperform the AoI-STO baseline (MSE $=4.24$). However, the advantage over the stronger AoI-only baseline (MSE $=2.90$) depends on the selected price: the calibrated value reduces MSE by approximately $4\%$, whereas a positive connectivity price alone does not guarantee an improvement. As described in Section~\ref{sec:results}, $\mu$ is calibrated using a single pilot run so that the connectivity reward accounts for $15\%$ of the per-slot objective. In this experiment, the resulting $\mu_{\mathrm{cal}}=0.8754$ is the best-performing value in the sweep, showing that the pilot-based rule selects this value without requiring a grid search for calibration.}

\begin{figure}[!t]
    \centering
    \includegraphics[width=0.92\columnwidth]{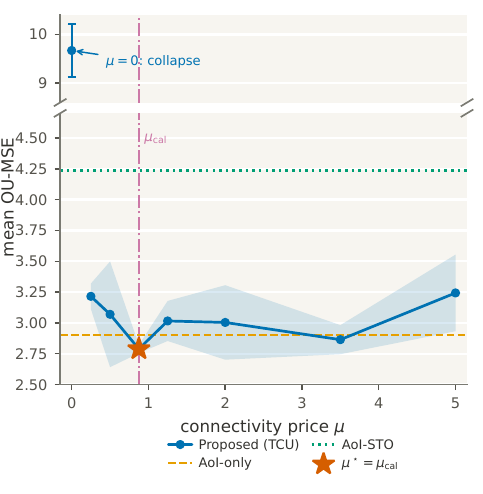}
    \caption{\revii{OU estimation MSE versus the connectivity price $\mu$ on the dispersed $L=1600$\,m Y geometry (averaged over three realizations; error bar for $\mu=0$ and shaded band for other $\mu$ values denote $\pm 1$ standard deviation). The AoI-only and AoI-STO baselines are shown as horizontal reference lines. The star marks the MSE-minimizing $\mu^\star$ and the dash-dotted line the calibrated $\mu_{\mathrm{cal}}=0.8754$.}}
    \label{fig:mu_sweep}
\end{figure}

\subsection{Optimality and Integer-Recovery Quality}
\label{subsec:opt_gap_recovery}

\rev{For $K=N$, the discrete assignment space contains only $N!$ permutations, making exhaustive assignment enumeration tractable for small systems. At $N=K=4$ we build an oracle by enumerating all $24$ permutations. For each permutation the oracle runs the full SCA continuous solve, applies the same repair and Step~4e certification as the deployed pipeline and records the certified per-slot objective. The comparison is therefore between deliverable operating points rather than raw matching weights. Over $T=30$ slots we report three statistics: the median relative gap, the fraction of slots within $0.5\%$ of the oracle and the exact-permutation match rate. The match statistic is strict: a slot counts as a match only when the returned permutation is identical to the oracle's, even if a different permutation attains the same objective.}

\begin{figure}[!t]
    \centering
    \includegraphics[width=0.9\columnwidth]{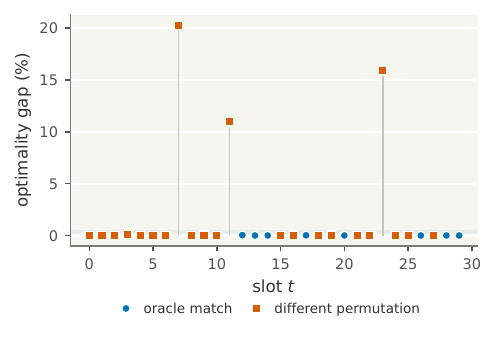}
    \caption{\rev{Per-slot optimality gap of the recovered assignment against exhaustive enumeration at $N=K=4$ ($T=30$). Markers distinguish slots where the recovered permutation matches the oracle from slots where a different permutation attains an equal or near-equal certified objective.}}
    \label{fig:assignment_quality}
\end{figure}

\rev{Fig.~\ref{fig:assignment_quality} reports the result. The median gap is $0.00\%$, and $27$ of $30$ slots lie within $0.5\%$ of the oracle. The exact-permutation match rate is only $26.7\%$. This low match rate results from ties. At $K=N$ several different pairings of UAVs to clusters often achieve exactly the same certified objective, so the oracle's best permutation is not unique. When the pipeline returns a different permutation with the same objective value, the match statistic counts a miss despite causing no performance loss. The three genuine misses reach gaps of $11$ to $20\%$, and the mean gap over all slots is $1.6\%$. These misses occur in isolated slots where a deep channel fade leaves the oracle's pairing as the only certifiable one. Hungarian matching returns the exact maximum-weight assignment at $\mathcal{O}(\max(K,N)^3)$ cost and is adopted for integer recovery throughout.}


\subsection{\rev{Convergence and }Complexity Analysis}
\label{subsec:complexity}

\begin{figure}[!t]
    \centering
    \includegraphics[width=0.8\columnwidth]{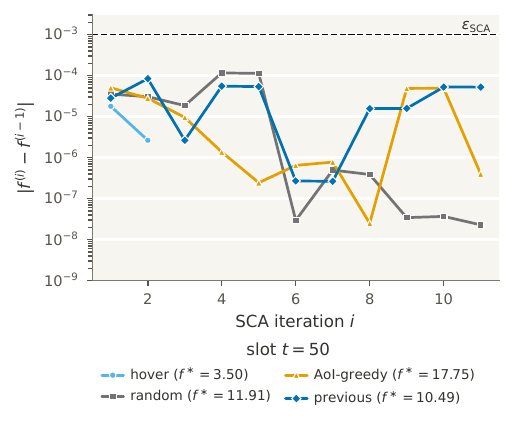}
    \caption{\rev{Per-iteration change of the surrogate objective for each warm start at a representative mid-mission slot, against the stopping tolerance $\varepsilon_{\mathrm{SCA}}$. The legend lists the converged objective of each start.}}
    \label{fig:convergence}
\end{figure}

The per-slot cost is dominated by the SCA inner loop. Each iterate solves a mixed-cone program with $n_{\mathrm{var}}=\mathcal{O}(KN+N^2)$ variables, $\mathcal{O}(KN)$ conic constraints and one $N\times N$ LMI from~\eqref{eq:C4}, yielding an interior-point cost of $\mathcal{O}(n_{\mathrm{var}}^{3.5}\log(1/\epsilon))$ plus the LMI contribution $\mathcal{O}(N^{6.5})$. With $J$ warm starts and at most $I_{\max}$ SCA iterations, the total per-slot complexity is
\[
\mathcal{O}\!\left(JI_{\max}\bigl[(KN)^{3.5}+N^{6.5}\bigr]\log(1/\epsilon)\right).
\]
Table~\ref{tab:complexity} summarizes the per-slot complexity and empirical
wall-time share. \rev{On an Apple M-series laptop with MOSEK and $4$ parallel missions, the proposed scheduler averages $9.9$\,s/slot at $N=4$, $K=15$ against $3.9$\,s/slot for BCD, a $2.5\times$ overhead for the joint optimization at its $48$-solve budget versus BCD's $15$. Convergence is fast in practice: Fig.~\ref{fig:convergence} shows the convergence behavior at a representative mid-mission slot using four initializations: hover, a random assignment, an AoI-greedy assignment, and the previous slot’s solution. For every initialization, the per-iteration change of the objective falls below $\varepsilon_{\mathrm{SCA}}=10^{-3}$ within one to three effective iterations, so most of the $I_{\max}=12$ budget serves as a safety margin. The converged objectives of the four starts differ by up to $5\times$. This sensitivity motivates the multistart strategy: although each run stabilizes quickly, the attained solution depends strongly on the initialization, and the algorithm retains the best certified solution.}

\begin{table}[!htbp]
\centering
\caption{Per-slot complexity and empirical wall-time share at $N=4$, $K=15$.}
\label{tab:complexity}
\footnotesize
\renewcommand{\arraystretch}{1.0}
\setlength{\tabcolsep}{4pt}
\begin{tabular}{@{}lcc@{}}
\toprule
Component & Complexity & Share \\
\midrule
SCA inner loop & $\mathcal{O}\!\left(JI_{\max}[(KN)^{3.5}+N^{6.5}]\log(1/\epsilon)\right)$ & $\sim\!97\%$ \\
Repair & $\mathcal{O}\!\left(N(KN)^{3.5}\log(1/\epsilon)\right)$ & $\sim\!2\%$ \\
State update & $\mathcal{O}(KN)$ & $<\!1\%$ \\
Hungarian matching & $\mathcal{O}(\max(K,N)^3)$ & $<\!0.1\%$ \\
\midrule
\textbf{Total} & $\mathcal{O}\!\left(JI_{\max}[(KN)^{3.5}+N^{6.5}]\log(1/\epsilon)\right)$ & \rev{$\sim\!9.9$\,s/slot} \\
\bottomrule
\end{tabular}
\end{table}

\subsection{Dynamic Scenario: UAV Failure}
\label{subsec:dynamic}

\rev{To assess behavior under mid-horizon swarm reconfiguration, we run the proposed scheduler on the nominal scenario over $T=100$ slots and remove one UAV at $t=50$, leaving three survivors to absorb the load. The planner is failure-aware: from the failure slot on, the lost UAV is excluded from the assignment, from the candidate set and from the connectivity certificate, which is evaluated on the surviving subgraph. The desired service rates $\omega_k$ are kept at their nominal values, so the deficit queue continues to operate against the original coverage target.}

\begin{figure}[!tb]
    \centering
    \includegraphics[width=0.8\columnwidth]{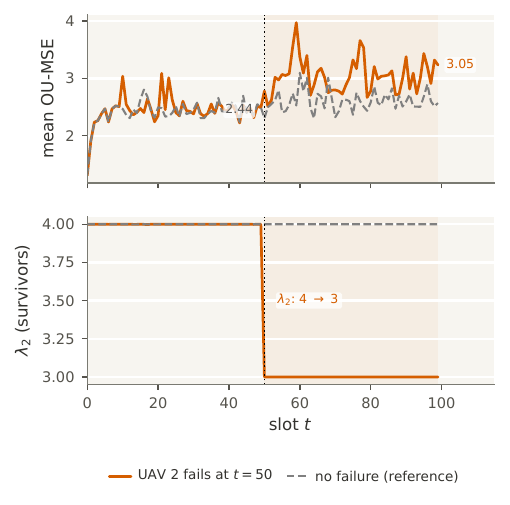}
    \caption{\rev{Mean OU estimation MSE (top) and expected connectivity of the surviving subgraph (bottom) as one of four UAVs is removed at $t=50$, against the no-failure reference.}}
    \label{fig:uav_failure_trajectory}
\end{figure}

\begin{figure}[!t]
    \centering
    \includegraphics[width=1\columnwidth,height=8cm]{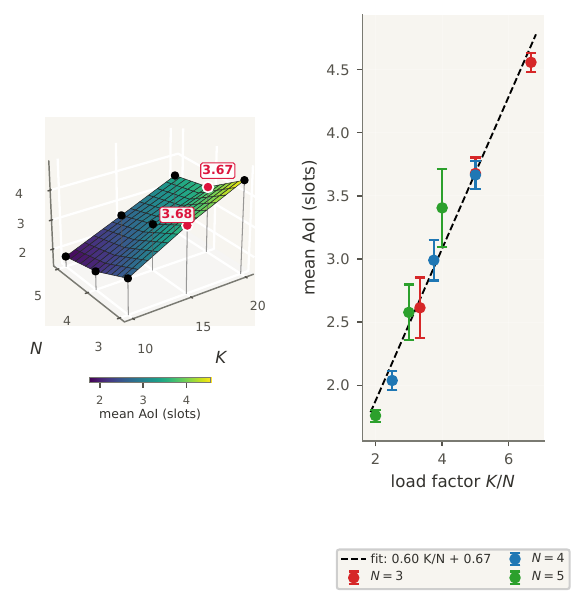}
    \caption{\rev{Scaling with swarm and cluster count. Left: steady-state mean AoI over $N\in\{3,4,5\}$ and $K\in\{10,15,20\}$. Markers show the nine measured configurations, and the two configurations at equal load $K/N=5$ are highlighted. Right: the same nine configurations against the load factor $K/N$ with $\pm 1$ std bars.}}
    \label{fig:nk_surface}
\end{figure}

\rev{Fig.~\ref{fig:uav_failure_trajectory} traces the transition. Before the failure, steady-state mean AoI is $2.93$. Losing one of the four UAVs removes a quarter of the service capacity and lifts mean AoI to a post-failure average of $3.82$. Over the final ten slots it drifts mildly to $4.03$, because the coverage target was set for four UAVs and the deficit queue now spreads it over three. The expected connectivity steps from $\lambda_2=4$ to $\lambda_2=3$, the value for the complete graph on the three survivors, and stays there. No slot ends in bailout and no realized outage occurs through the transition. The surviving UAVs take over the failed UAV's clusters: as these clusters accumulate service deficits, their scheduling priorities increase, enabling the handover without any cluster reaching its starvation threshold.}


\subsection{\rev{Scaling with Swarm and Cluster Count}}
\label{subsec:scaling}
\rev{To characterize performance beyond the studied operating point, we sweep $N\in\{3,4,5\}$ against $K\in\{10,15,20\}$, yielding nine configurations. Each configuration is evaluated over $T=100$ slots and $4$ seeds. The sweep stays inside the certified scope of Proposition~\ref{prop:chance_connectivity}: the per-edge floor follows the union bound, with $\eta=0.05$ at $N\in\{3,4\}$ and $\eta=0.03$ at $N=5$. Fig.~\ref{fig:nk_surface} shows that the steady-state mean AoI varies monotonically with both $K$ and $N$: it increases as the number of clusters $K$ grows and decreases as the number of UAVs $N$ increases. Across all nine configurations, the mean AoI is determined primarily by the ratio $K/N$, which represents the per-slot assignment load. When plotted against $K/N$, the results from all configurations fall on a single line. Configurations with equal load are statistically indistinguishable: at $K/N=5$, the configuration $N=3$, $K=15$ gives $3.68\pm0.13$ and the configuration $N=4$, $K=20$ gives $3.67\pm0.11$. These results support the selected operating point and indicate that information freshness is governed by the assignment load rather than by the absolute fleet size. The certified success rate is at least $0.997$ in every configuration. A single realized outage occurs across the roughly ten thousand attempts.}

\section{Conclusion}
\label{sec:conclusion}

This paper \rev{develops a unified framework for jointly optimizing remote-estimation accuracy}, finite-blocklength link reliability and probabilistic swarm connectivity in UAV-swarm IoT data collection. \rev{The finite-horizon formulation minimizes the accumulated estimation error of the monitored processes, offset by a weighted connectivity reward,} over cluster--UAV assignment, UAV motion, time allocation and transmit power, subject to finite-blocklength reliability, probabilistic connectivity and energy constraints. The resulting nonconvex program is solved per slot by successive convex approximation with Hungarian integer recovery and post-recovery verification.

\rev{Across the nominal and load-imbalanced deployments, the proposed scheduler attains the lowest estimation error and the highest service success rate while matching the mean AoI of the best AoI heuristic. No delivered update violates the reliability target in more than $19{,}000$ certified service attempts. The median assignment optimality gap against exhaustive enumeration is $0.00\%$. In a dispersed deployment where the connectivity constraint is binding, removing the connectivity reward allows the swarm topology to deteriorate and increases the mean AoI by a factor of $4.8$. The hard connectivity constraint prevents connectivity-infeasible actions from being executed, whereas the connectivity reward maintains the operating margin required for sustained service.}

The main limitation is computational complexity: with an average runtime of approximately \rev{$10$\,s} per slot, the proposed scheduler is more suitable as an offline planner or labeling oracle than as a real-time controller. Future work will focus on amortizing this cost through imitation or deep reinforcement learning and tightening the chance-connectivity certificate using matrix-concentration bounds on the expected Laplacian\rev{, thereby improving scalability while preserving probabilistic guarantees}.

\bibliographystyle{IEEEtran}
\bibliography{tcu-references}

\end{document}